\documentclass[letterpaper, 11pt]{article}
\usepackage{fullpage}
\usepackage{graphicx}
\usepackage{amsmath, amsthm, amssymb, cite, color, comment, mathtools, url}
\usepackage[unicode]{hyperref}

\newtheorem{theorem}{Theorem}[section]
\newtheorem{lemma}[theorem]{Lemma}

\newtheorem{claim}[theorem]{Claim}

\theoremstyle{definition}
\newtheorem{definition}[theorem]{Definition}

\newtheorem*{problem}{Problem}
\newtheorem*{algorithm*}{Algorithm}

\theoremstyle{remark}
\newtheorem*{remark}{Remark}

\newcommand{\argmin}{\mathop{\rm arg\,min}\limits}
\newcommand{\id}{1_\Gamma}
\newcommand{\RR}{{\mathbb R}}
\newcommand{\ZZ}{{\mathbb Z}}

\newcommand{\cS}{{\cal S}}

\newcommand{\tf}{{\tilde{f}}}

\newcommand{\tE}{{\tilde{E}}}
\newcommand{\tF}{{\tilde{F}}}
\newcommand{\tG}{{\tilde{G}}}
\newcommand{\tP}{{\tilde{P}}}
\newcommand{\tQ}{{\tilde{Q}}}
\newcommand{\tR}{{\tilde{R}}}
\newcommand{\tT}{{\tilde{T}}}
\newcommand{\tV}{{\tilde{V}}}
\newcommand{\tell}{{\tilde{\ell}}}
\newcommand{\comp}{{\mathrm{comp}}}
\newcommand{\dist}{{\mathrm{dist}}}
\newcommand{\depth}{{\mathrm{depth}}}
\newcommand{\pred}{{\mathrm{pred}}}
\newcommand{\parent}{{\mathrm{parent}}}
\renewcommand{\root}{{\mathrm{root}}}

\title{Finding a Shortest~Non-zero~Path in~Group-Labeled~Graphs\thanks{A preliminary version~\cite{Yamaguchi2020} of this paper appeared in SODA 2020.}}
\author{
  Yoichi Iwata\thanks{AtCoder Inc., Tokyo 160-0022, Japan. Email: \texttt{wata@atcoder.jp}} \and
  Yutaro Yamaguchi\thanks{Osaka University, Osaka 565-0871, Japan. Email: \texttt{yutaro.yamaguchi@ist.osaka-u.ac.jp}}}
\date{\empty}

\begin{document}
\maketitle
\thispagestyle{empty}

\begin{abstract}
We study a constrained shortest path problem in group-labeled graphs with nonnegative edge length, called the \emph{shortest non-zero path problem}.
Depending on the group in question, this problem includes two types of tractable variants in undirected graphs: one is the parity-constrained shortest path/cycle problem, and the other is computing a shortest noncontractible cycle in surface-embedded graphs.

For the shortest non-zero path problem with respect to finite abelian groups, Kobayashi and Toyooka (2017) proposed a randomized, pseudopolynomial-time algorithm via permanent computation.
For a slightly more general class of groups, Yamaguchi (2016) showed a reduction of the problem to the weighted linear matroid parity problem.
In particular, some cases are solved in strongly polynomial time via the reduction with the aid of a deterministic, polynomial-time algorithm for the weighted linear matroid parity problem developed by Iwata and Kobayashi (2021), which generalizes a well-known fact that the parity-constrained shortest path problem is solved via weighted matching.

In this paper, as the first general solution independent of the group, we present a rather simple, deterministic, and strongly polynomial-time algorithm for the shortest non-zero path problem.
This result captures a common tractable feature behind the parity and topological constraints in the shortest path/cycle problem.
The algorithm is based on Dijkstra's algorithm for the unconstrained shortest path problem and Edmonds' blossom shrinking technique in matching algorithms;
this approach is inspired by Derigs' faster algorithm (1985) for the parity-constrained shortest path problem via a reduction to weighted matching.
Furthermore, we improve our algorithm so that it does not require explicit blossom shrinking, and make the computational time match Derigs' one.
In the speeding-up step, a dual linear programming formulation of the equivalent problem based on potential maximization for the unconstrained shortest path problem plays a key role.
\end{abstract}


\clearpage
\thispagestyle{empty}
\tableofcontents
\clearpage
\setcounter{page}{1}

\section{Introduction}
Finding a shortest path between two specified vertices, say $s$ and $t$, is a fundamental task in graphs.
As a variant, it is well-known that one can find a shortest \emph{odd} (or \emph{even}) $s$--$t$ path in an undirected graph with nonnegative edge length in strongly polynomial time via a reduction to the weighted matching problem (see, e.g., \cite[\S~29.11e]{Schrijver2003}), where ``odd'' (or ``even'') designates the parity of the number of traversed edges.
We remark that, in the directed case, even determining whether a given directed graph contains an odd (or even) directed path from $s$ to $t$ or not is NP-complete~\cite{LP}.

A shortest cycle is also closely related, as it can be found, at least, by computing a shortest $s$--$t$ path in the graph obtained by removing each edge $e = \{s, t\}$.
Shortest \emph{noncontractible} cycles in graphs embedded in surfaces have been studied in topological graph theory with several motivations; e.g., in the unweighted case, the minimum number of edges in a noncontractible cycle is an index of embeddings called the \emph{edge-width}.
The first polynomial-time algorithm for finding a shortest noncontractible cycle is based on a simple observation given by Thomassen~\cite{Thomassen1990}, the so-called \emph{$3$-path condition}.
The current fastest algorithm was proposed by Erickson and Har-Peled~\cite{EH2004}, and there are several faster ones for bounded-genus cases~\cite{CCE2013, Fox2013}.
For more detailed literature, we refer the readers to \cite{Colin2017, Erickson2012}.

As a common generalization of the parity and topological conditions (or others in some contexts), paths and cycles with label conditions in group-labeled graphs have recently been investigated from both combinatorial and algorithmic points of view \cite{CGGGLS2006, CCG2008, Huynh, HJW2019, KW2006, KKY, LRS2017, TY2016, Wollan2011, Yamaguchi2016_SIDMA}, where a group-labeled graph is a directed graph with each arc labeled by an element of a fixed group.
Formally, for a group $\Gamma$, a \emph{$\Gamma$-labeled graph} is a pair of a directed graph and a mapping from the arc set to $\Gamma$.
In a $\Gamma$-labeled graph, the label of a walk is defined by sequential applications of the group operation of $\Gamma$ to the labels of the traversed arcs, where each arc (e.g., from $u$ to $v$) can be traversed in the backward direction (from $v$ to $u$) by inverting its label.
A walk is said to be \emph{non-zero} if its label is not equal to the identity element $\id$ of $\Gamma$. See Section~\ref{sec:pre} for precise definitions.

On one hand, the parity condition in undirected graphs is handled by choosing $\Gamma = \ZZ_2 = \ZZ / 2\ZZ = (\{0, 1\}, +)$, orienting each edge arbitrarily, and assigning the label $1 \in \ZZ_2$ to all the resulting arcs.
Then, the label of a walk is non-zero if and only if the number of traversed edges is odd.
Conversely, the label in any $\ZZ_2$-labeled graph can be regarded as the parity by subdividing each arc with label $0 \in \ZZ_2$ as two arcs with label $1 \in \ZZ_2$.

On the other hand, contractibility in surfaces can be represented as follows.
Suppose that an undirected graph $G = (V, E)$ is embedded in a surface $S$.
Fix any point $x$ in $S$, and let $\Gamma = \pi_1(S, x)$ be the fundamental group of $S$ at the basepoint $x$, in which each element is a homotopy class of closed curves with endpoint $x$ and the group operation corresponds to composition of such curves.
We also fix any simple curve $\gamma_v$ in $S$ from $x$ to each vertex $v \in V$.
Then, after orienting each edge arbitrarily, we can define the label of each resulting arc $\vec{e} = uv$ as the homotopy class that contains the closed curve composed of $\gamma_u$, $\vec{e}$, and $\overline{\gamma_v}$ (the reverse of $\gamma_v$), so that the label of a closed walk in $G$ is non-zero if and only if it is noncontractible in $S$.

\medskip
In this paper, we focus on the \emph{shortest non-zero path problem}: given a $\Gamma$-labeled graph with two specified vertices $s$ and $t$ and a nonnegative length of each edge in the underlying graph, we are required to find a shortest non-zero $s$--$t$ path.
Note that any element $\alpha \in \Gamma$ can be chosen as the forbidden label instead of $\id$, e.g., by adding to the input graph a new vertex $t'$ (as the end vertex instead of $t$) and a new arc from $t$ to $t'$ with label $\alpha^{-1}$.
Thus, a shortest non-zero cycle is also found, at least, by solving this problem for the graph obtained by removing each arc.

For the case when $\Gamma$ is finite and abelian, Kobayashi and Toyooka~\cite{KT2017} proposed a randomized, pseudopolynomial-time algorithm via permanent computation, where ``pseudopolynomial'' means that the computational time polynomially depends on the length values (i.e., if every edge has a unit length, then it is bounded by a polynomial in the graph size).
For a slightly more general case when $\Gamma$ is finitely generated and abelian (or dihedral, etc.), the second author~\cite{Yamaguchi2016_ISAAC} showed that a generalized problem reduces to the weighted linear matroid parity problem, for which Iwata and Kobayashi~\cite{IK2021} devised a deterministic, polynomial-time algorithm for the case when the field in question is finite or that of rational numbers.
In particular, for any prime $p$, the shortest non-zero path problem in $\ZZ_p$-labeled graphs is solved in strongly polynomial time via the reduction, which generalizes the aforementioned fact that the shortest odd (even) path problem is solved via weighted matching.
However, since the fundamental group $\pi_1(S, x)$ is infinite or non-abelian even when $S$ is the torus or the Klein bottle, respectively,
these results do not explain combinatorial tractability of topological constraints well.

In this paper, we present a deterministic, strongly polynomial-time algorithm for the shortest non-zero path problem in general, where basic operations on the group $\Gamma$ (such as the group operation, identity test of two elements, and getting the inverse element) as well as arithmetics on edge length (e.g., for reals) are regarded as elementary (i.e., performed in constant time).
This implies that shortest odd (even) paths/cycles and shortest noncontractible cycles enjoy a common tractable feature that is captured as one induced by groups.

\begin{theorem}\label{thm:main}
There exists a deterministic algorithm for the shortest non-zero path problem that requires ${\rm O}(m \log n)$ elementary operations,
where $n$ and $m$ denote the numbers of vertices and of arcs (edges), respectively, in the input graph, which is assumed to be connected.
\end{theorem}

Our result is inspired by and essentially extends an $\mathrm{O}(m \log n)$-time algorithm for the shortest odd (even) path problem given by Derigs~\cite{Derigs1985}, but its extendability is quite nontrivial.
Derigs' algorithm utilizes the reduction to weighted matching; more specifically, it computes a shortest augmenting $s$--$t$ path in an auxiliary two-layered graph\footnote{Each layer is a copy of the original input graph $G$ (intuitively, one is regarded as odd and the other is even), and two vertices in the different layers are adjacent if and only if they are originated by the same vertex in $G$.} with respect to a nearly perfect matching between the two layers that only exposes $s$ and $t$.
This approach is highly specialized to the case of $|\Gamma| = 2$ from the viewpoint of our problem, and a counterpart of such an auxiliary graph seems hopeless for any other case even when $|\Gamma| = 3$.

In order to employ a similar idea directly in the original input graph, we slightly reformulate the shortest non-zero path problem as follows.
We first compute a shortest-path tree $T$ rooted at $s$ by Dijkstra's algorithm~\cite{Dijkstra} (by ignoring the constraint).
If a unique $s$--$t$ path $P_t$ in $T$ is non-zero, then we conclude that $P_t$ is a desired solution.
Otherwise, the remaining task is to find an $s$--$t$ path that is shortest among those whose labels are different from $P_t$, which we call the \emph{shortest unorthodox path problem}.

As a key concept to solve this problem, we introduce a \emph{lowest blossom}\footnote{This blossom corresponds to a usual blossom after the reduction to weighted matching when $|\Gamma| = 2$.} with respect to $T$, which is a non-zero cycle $C$ enjoying the following properties.
If $t$ is on $C$, then we can assure that the $s$--$t$ path obtained from $P_t$ by detouring around $C$ (which must change the label) is a desired solution.
Otherwise, we can obtain a shortest unorthodox $s$--$t$ path by expanding a shortest unorthodox $s$--$t$ path in a small graph after shrinking $C$ into a single vertex, which is found recursively.
In addition, we can find a lowest blossom in $\mathrm{O}(m)$ time per one.
Since the number of vertices must decrease by shrinking, we thus obtain a simple $\mathrm{O}(nm)$-time algorithm.

While the recursive algorithm is intuitive and easy to explain, we have to avoid such an explicit recursion in order to make our algorithm match Derigs' one in running time.
For this purpose, we provide an alternative $\mathrm{O}(nm)$-time algorithm without explicit blossom shrinking, and then improve its running time bound to ${\rm O}(m \log n)$ with the aid of basic data structures.
We also introduce a dual linear programming (LP) formulation of the shortest unorthodox path problem, which is based on potential maximization for the unconstrained shortest path problem.
Interestingly, the correctness of the alternative algorithm and the LP formulation is shown simultaneously.

As a byproduct of introducing the concept of lowest blossoms, we also obtain a faster algorithm for finding a shortest non-zero cycle than the na\"ive approach by solving the shortest non-zero path problem for each arc, which requires ${\rm O}(m^2 \log n)$ time in total.
The idea is similar to the algorithm for finding a noncontractible cycle proposed in \cite{EH2004}.
We see that each lowest blossom corresponds to a shortest non-zero closed walk with end vertex $s$.
Thus, by the nonnegativity of edge length, for finding a shortest non-zero cycle, it suffices to compute one lowest blossom with respect to one shortest-path tree rooted at each vertex.
This can be done in linear time per vertex after computing a shortest-path tree by Dijkstra's algorithm, and a bottleneck is an $\mathrm{O}(m + n \log n)$-time implementation of Dijkstra's algorithm with the aid of Fibonacci heaps (see, e.g., \cite[$\S$~7.4]{Schrijver2003}).

\begin{theorem}\label{thm:cycle}
There exists a deterministic algorithm for finding a shortest non-zero cycle in a  $\Gamma$-labeled graph with nonnegative edge length
that requires ${\rm O}(n(m + n\log n))$ elementary operations.
\end{theorem}

The rest of this paper is organized as follows.
In Section~\ref{sec:pre}, we define terms and notations, and sketch a basic strategy.
In Section~\ref{sec:CUC}, as a key concept in our algorithms, we introduce \emph{lowest blossoms} and several operations.
In Section~\ref{sec:algorithm}, we present an ${\rm O}(nm)$-time recursive algorithm based on the blossom shrinking idea.
Finally, in Section~\ref{sec:fast}, we show an ${\rm O}(m \log n)$-time algorithm together with a dual LP formulation, which complete the proof of Theorem~\ref{thm:main}.

\section{Preliminaries}\label{sec:pre}
Let $\Gamma$ be a group, which can be non-abelian or infinite.
We adopt the multiplicative notation for $\Gamma$ with denoting the identity element by $\id$.

A \emph{$\Gamma$-labeled graph} is a pair $(\vec{G}, \psi)$ of a directed graph $\vec{G} = (V, \vec{E})$ and a mapping $\psi \colon \vec{E} \to \Gamma$.
We denote by $G = (V, E)$ the underlying graph of $\vec{G}$, i.e., $E \coloneqq \{\, e = \{u, v\} \mid \vec{e} = uv \in \vec{E} \,\}$,
and define $\psi_G(e, uv) \coloneqq \psi(\vec{e})$ and $\psi_G(e, vu) \coloneqq \psi(\vec{e})^{-1}$ for each edge $e = \{u, v\} \in E$ with the corresponding arc $\vec{e} = uv \in \vec{E}$.
For simple notation, we deal with a $\Gamma$-labeled graph $(\vec{G}, \psi)$ as its underlying graph $G$ including the information of $\psi_G$ defined above.
We assume that $G$ has no loop but may have parallel edges, i.e., $E$ is a multiset of two-element subsets of $V$.
We often refer to a subgraph $(V, F)$ of $G$ as its edge set $F \subseteq E$.
For a vertex $v \in V$, we define $\delta_G(v) \coloneqq \{\, e \mid v \in e \in E \,\}$.

A {\em walk} in a $\Gamma$-labeled graph $G = (V, E)$ is an alternating sequence of vertices and edges,
$W = (v_0, e_1, v_1, \ldots, e_k, v_k)$, such that $e_i = \{v_{i-1}, v_i\} \in E$ for each $i = 1, 2, \ldots, k$.
We often call $W$ a \emph{$v_0$--$v_k$ walk} by specifying its end vertices $v_0$ and $v_k$.
We say that $W$ is \emph{closed} when $v_0 = v_k$.
We define $V(W) \coloneqq \{v_0, v_1, \dots, v_k\}$ and $E(W) \coloneqq \{e_1, e_2, \dots, e_k\}$,
where the multiplicity is ignored for $V(W)$ (as a set) and is included for $E(W)$ (as a multiset).
For $i, j$ with $0 \leq i \leq j \leq k$,
let $W[v_i, v_j]$ denote the subwalk $(v_i, e_{i+1}, v_{i+1}, \dots, e_j, v_j)$ of $W$ (when $v_i$ or $v_j$ appears multiple times in $W$, we specify which one is chosen).
We call $W$ a {\em path} if all vertices $v_i$ are distinct (i.e., $|V(W)| = k + 1$),
and a {\em cycle} if $v_0 = v_k$ and $W[v_1, v_k]$ is a path with $e_1 \not\in E(W[v_1, v_k])$. %(in particular, two parallel edges form a cycle).
Let $\overline{W}$ denote the reversed walk of $W$, i.e., $\overline{W} = (v_k, e_k, \dots, v_1, e_1, v_0)$.
For a walk $W' = (v'_0, e'_1, v'_1, \dots, e'_l, v'_l)$ with $v'_0 = v_k$,
we define the concatenation of $W$ and $W'$ as $W * W' \coloneqq (v_0, e_1, v_1, \dots, e_k, v_k = v'_0, e'_1, v'_1, \dots, e'_l, v'_l)$.

For given edge lengths $\ell \in \RR_{\geq 0}^E$,
the {\em length} of $W$ is defined as $\ell(W) \coloneqq \sum_{e \in E(W)} \ell(e) = \sum_{i=1}^k \ell(e_i)$.
We say that a walk $W$ is \emph{shortest} if $\ell(W)$ is minimized under a specified constraint, e.g., $W$ is an $s$--$t$ path in $G$ for some fixed vertices $s, t \in V$.
The {\em label} of $W$ is defined as 
\[\psi_G(W) \coloneqq \psi_G(e_1, v_0v_1)\cdot\psi_G(e_2, v_1v_2)\cdot \cdots \cdot \psi_G(e_k,  v_{k-1}v_k),\]
where we remark again that the group $\Gamma$ can be non-abelian.
By definition, we always have $\psi_G(\overline{W}) = \psi_G(W)^{-1}$ and $\psi_G(W * W') = \psi_G(W) \cdot \psi_G(W')$.
A walk $W$ in $G$ is \emph{non-zero} if $\psi_G(W) \neq \id$.
Note that whether cycles are non-zero or not is invariant under the choices of direction and end vertices, because $\psi_G(\overline{C}) = \psi_G(C)^{-1}$ and $\psi_G(C[v_1, v_k] * (v_0, e_1, v_1)) = \psi_G(e_1, v_0v_1)^{-1}\cdot\psi_G(C)\cdot\psi_G(e_1, v_0v_1)$ hold for any cycle $C = (v_0, e_1, v_1, e_2, v_2, \dots, e_k, v_k = v_0)$.

We are now ready to state our problem formally.

\begin{problem}[Shortest Non-zero Path]
\begin{description}
  \setlength{\itemsep}{0mm}
\item[]

\item[Input:]
  A $\Gamma$-labeled graph $G = (V, E)$, edge lengths $\ell \in \RR_{\geq 0}^{E}$, and two distinct vertices $s, t \in V$.

\item[Goal:]
  Find an $s$--$t$ path $P$ in $G$ minimizing $\ell(P)$ subject to $\psi_G(P) \neq \id$ (if any).
\end{description}
\end{problem}

Without loss of generality, we assume that the input graph is connected.
In order to solve this problem,
we first compute a shortest-path tree $T$ of $(G, \ell)$ rooted at $s$ (formally defined in Definition~\ref{def:tree}) by Dijkstra's algorithm.
If, fortunately, a unique $s$--$t$ path $P_t$ in $T$ is non-zero, then we have to do nothing else.
Otherwise, an $s$--$t$ path $Q$ in $G$ is non-zero if and only if $\psi_G(Q) \neq \psi_G(P_t) \ ( = \id)$.
When $T$ is fixed, we say that an $s$--$t$ path $Q$ is \emph{unorthodox} if $\psi_G(Q) \neq \psi_G(P_t) \eqqcolon \psi_T(t)$.
Thus, our main task is to solve the following problem.

\begin{problem}[Shortest Unorthodox Path]
\begin{description}
  \setlength{\itemsep}{0mm}
\item[]

\item[Input:]
  A connected $\Gamma$-labeled graph $G = (V, E)$, edge lengths $\ell \in \RR_{\geq 0}^{E}$, two distinct vertices $s, t \in V$, and a shortest-path tree $T$ of $(G, \ell)$ rooted at $s$.

\item[Goal:]
  Find an $s$--$t$ path $Q$ in $G$ minimizing $\ell(Q)$ subject to $\psi_G(Q) \neq \psi_T(t)$ (if any), or conclude that all $s$--$t$ paths in $G$ are of label $\psi_T(t)$.
\end{description}
\end{problem}

We formally define a shortest-path tree and related concepts as follows.

\begin{definition}[Shortest-Path Trees]\label{def:tree}
Let $G = (V, E)$ be a connected $\Gamma$-labeled graph with edge lengths $\ell \in \RR_{\geq 0}^E$ and a specified vertex $s \in V$.
A spanning tree $T \subseteq E$ is called a \emph{shortest-path tree of $(G, \ell)$ rooted at $s$} (or an \emph{$s$-SPT} for short) if, for each $v \in V$, a unique $s$--$v$ path in $T$ is shortest in $G$.
When an $s$-SPT $T$ of $(G, \ell)$ is fixed, for each $v \in V$, we denote by $P_v$ the unique $s$--$v$ path in $T$, and define $\dist_T(v) \coloneqq \ell(P_v)$ and $\psi_T(v) \coloneqq \psi_G(P_v)$.
In addition, an edge $e = \{u, v\} \in E$ is said to be \emph{consistent (with $T$)} if $\psi_T(u) \cdot \psi_G(e, uv) = \psi_T(v)$, and \emph{inconsistent} otherwise.
Similarly, an $x$--$y$ path $R$ in $G$ is said to be \emph{orthodox} if $\psi_T(x) \cdot \psi_G(R) = \psi_T(y)$, and \emph{unorthodox} otherwise.\footnote{From the basic facts shown in \cite[\S~2.3]{KKY}, when $\psi_T(t) = \id$, we can assume that $\psi_T(v) = \id$ for all $v \in V$ by shifting the labels of edges around vertices in $V \setminus \{s, t\}$ if necessary. Hence, ``inconsistent'' and ``unorthodox'' can be translated into ``non-zero'' (when we want to solve the shortest unorthodox path problem). In addition, one can test whether a given instance is feasible (i.e., has a non-zero $s$--$t$ path) or not in linear time in advance.}
\end{definition}

A shortest-path tree rooted at $s$ can be computed by Dijkstra's algorithm with a slight adjustment to our group-labeled setting as follows.

\begin{algorithm*}[{{\sc Dijkstra}$[G, \ell, s]$}]
\begin{description}
  \setlength{\itemsep}{0mm}
\item[]
\item[Input:]
  A connected $\Gamma$-labeled graph $G = (V, E)$, edge lengths $\ell \in \RR_{\geq 0}^E$, and a vertex $s \in V$.

\item[Output:]
  A shortest-path tree $T$ of $(G, \ell)$ rooted at $s$ with $\dist_T$ and $\psi_T$.

\item[Step 1.]\vspace{1mm}
  For each $v \in V \setminus \{s\}$, set $\dist_T(v) \leftarrow +\infty$.
  Set $\dist_T(s) \leftarrow 0$, $\psi_T(s) \leftarrow \id$, and $U \leftarrow \emptyset$.

\item[Step 2.]
  While $\Delta \coloneqq \min\left\{\, \dist_T(v) \mid v \in V \setminus U \,\right\} < +\infty$, pick $v \in V \setminus U$ with $\dist_T(v) = \Delta$,
  update $U \leftarrow U \cup \{v\}$, and do the following for each edge $e = \{v, w\} \in \delta_G(v)$ with $w \not\in U$:
  if $\dist_T(w) > \Delta + \ell(e)$, then update $\dist_T(w) \leftarrow \Delta + \ell(e)$, $\psi_T(w) \leftarrow \psi_T(v) \cdot \psi_G(e, vw)$, and $\mathrm{last}(w) \leftarrow e$.
  
\item[Step 3.]
  Set $T \leftarrow \{\, \mathrm{last}(v) \mid v \in V \setminus \{s\} \,\}$, and return $T$ with $\dist_T$ and $\psi_T$.
\end{description}
\end{algorithm*}

Finally, we observe useful properties on shortest-path trees.

\begin{lemma}\label{lem:tree}
For any shortest-path tree $T$ of $(G, \ell)$ rooted at $s$, the following properties hold.
\begin{itemize}
  \setlength{\itemsep}{.5mm}
\item[$(1)$]
  For any vertices $x, y \in V$ and any $x$--$y$ walk $R$ in $G$, we have $\ell(R) \geq \left|\dist_T(x) - \dist_T(y)\right|$.
\item[$(2)$]
  For any vertices $x, y \in V$, any unorthodox $x$--$y$ path $R$ in $G$ traverses some inconsistent edge in $E \setminus T$.
\end{itemize}
\end{lemma}

\begin{proof}
\begin{comment}
$(1)$~
For each edge $e = \{u, v\} \in E$, since $P_u$ and $P_v$ are both shortest in $G$,
we have $\ell(P_v) \leq \ell(P_u) + \ell(e)$ and $\ell(P_u) \leq \ell(P_v) + \ell(e)$,
and hence $\ell(e) \geq \left|\ell(P_u) - \ell(P_v)\right| = \left|\dist_T(u) - \dist_T(v)\right|$.
Thus,
\begin{align*}
\ell(R) = \sum_{e \in E(R)}\ell(e)
\geq \sum_{e = \{u, v\} \in E(R)} \left|\dist_T(u) - \dist_T(v)\right|
\geq \left|\dist_T(x) - \dist_T(y)\right|.
\end{align*}
\end{comment}
$(1)$~
As $P_x * R$ is an $s$--$y$ walk and $P_y * \overline{R}$ is an $s$--$x$ walk, by the nonnegativity of $\ell$, we have $\ell(P_y) \leq \ell(P_x) + \ell(R)$ and $\ell(P_x) \leq \ell(P_y) + \ell(R)$.

\medskip
\noindent
$(2)$~
Fix two vertices $x, y \in V$ and an $x$--$y$ path $R = (x = v_0, e_1, v_1, \dots, e_k, v_k = y)$.
If every edge $e = \{u, v\} \in E(R)$ is consistent with $T$,
then
\begin{align*}
\psi_G(R) &= \psi_G(e_1, v_0v_1) \cdot \psi_G(e_2, v_1v_2) \cdot \cdots \cdot \psi_G(e_k, v_{k-1}v_k)\\
 &= \psi_T(v_0)^{-1} \cdot \psi_T(v_1) \cdot \psi_T(v_1)^{-1} \cdot \cdots \cdot \psi_T(v_k)\\
 &= \psi_T(x)^{-1} \cdot \psi_T(y).
\end{align*}
Hence, if $\psi_T(x) \cdot \psi_G(R) \neq \psi_T(y)$, then there exists an inconsistent edge $e = \{u, v\} \in E(R)$.
Moreover, since every edge in $T$ is consistent with $T$ by definition, we conclude that $e \in E(R) \setminus T$.
\end{proof}

\section{Lowest Blossoms}\label{sec:CUC}
In this section, we introduce the concept of lowest blossoms and several operations utilized in our recursive algorithm shown in Section~\ref{sec:algorithm}.
Throughout this section, let $G = (V, E)$ be a connected $\Gamma$-labeled graph with edge lengths $\ell \in \RR_{\geq 0}^E$ and two distinct vertices $s, t \in V$, and fix a shortest-path tree $T$ of $(G, \ell)$ rooted at $s$.
By Lemma~\ref{lem:tree}(2), if every edge $e = \{u, v\} \in E \setminus T$ is consistent with $T$,
then we can immediately conclude that all $s$--$t$ paths in $G$ are orthodox, i.e., of label $\psi_T(t)$.
We assume that there exists an inconsistent edge, and define \emph{blossoms} as follows.

\begin{definition}[Lowest Blossoms]\label{def:canonical}
For an $s$-SPT $T$ of $(G, \ell)$ and an inconsistent edge $e = \{u, v\} \in E \setminus T$,
let $W_{e} \coloneqq P_u * (u, e, v) * \overline{P_v}$ (where $u$ and $v$ are symmetric as $e$ has no direction, and we fix an arbitrary direction to define $W_e$).
Then, $W_e$ contains a unique cycle $C_e$ as its subwalk, and is decomposed as $W_{e} = P_{b_e} * C_{e} * \overline{P_{b_e}}$ (see Fig.~\ref{fig:C_e}).
We call $C_e$ a \emph{blossom}, and $b_e$ and $P_{b_e}$ its \emph{base} and \emph{stem}, respectively.
The \emph{height} of $C_e$ is defined as $\frac{1}{2}\ell(W_e) = \frac{1}{2}\left(\dist_T(u) + \dist_T(v) + \ell(e)\right)$, and a blossom of the minimum height is said to be \emph{lowest}.
\end{definition}

\begin{figure}[tb]%\vspace{-1mm}
 \begin{center}
  \includegraphics[scale=0.8]{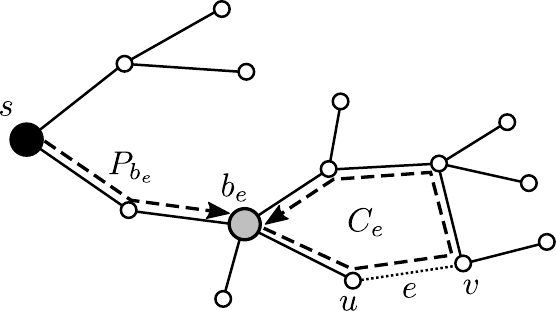}
 \end{center}\vspace{-3mm}
 \caption{The blossom $C_e$ with base $b_e$ and stem $P_{b_e}$ for an inconsistent edge $e = \{u, v\} \in E \setminus T$, where solid lines represent edges in $T$.}\vspace{-1mm}
 \label{fig:C_e}
\end{figure}

Note that any blossom $C_e$ is a non-zero cycle (which may consist of two parallel edges) as the edge $e$ is inconsistent with $T$, i.e., $\psi_T(u) \cdot \psi_G(e, uv) \neq \psi_T(v)$.
We observe the following property of lowest blossoms, which will be used to prove Theorem~\ref{thm:cycle}.

\begin{lemma}\label{lem:lb}
Let $T$ be an $s$-SPT of $(G, \ell)$,
and $C_e$ a lowest blossom with base $b_e$ for $e = \{u, v\} \in E \setminus T$.
Then, $W_{e} = P_{b_e} * C_{e} * \overline{P_{b_e}}$ is a shortest non-zero closed walk with end vertex $s$ in $G$.
\end{lemma}

\begin{proof}
Let $W$ be an arbitrary non-zero closed walk with end vertex $s$ in $G$, and we show $\ell(W) \geq \ell(W_e)$.
By Lemma~\ref{lem:tree}(2), $W$ traverses some inconsistent edge $f = \{x, y\} \in E \setminus T$, and define $W_f \coloneqq P_x * (x, f, y) * \overline{P_y}$.
Since $C_e$ is a lowest blossom, we have $\ell(W_e) \leq \ell(W_f)$.
In addition, by Lemma~\ref{lem:tree}(1), we have $\ell(W_f) \leq \ell(W[s, x]) + \ell(f) + \ell(W[y, s]) = \ell(W)$, and we are done.
\end{proof}

The following lemma shows one of the two key properties of lowest blossoms:
for any vertex $w$ on a lowest blossom $C_e$ except for the base $b_e$,
a unique detour around $C_e$ from $P_w$ yields a shortest unorthodox $s$--$w$ path $Q_w$.

\begin{lemma}\label{lem:second}
Let $T$ be an $s$-SPT of $(G, \ell)$,
and $C_e$ a lowest blossom with base $b_e$ for $e = \{u, v\} \in E \setminus T$.
Then, for any vertex $w \in V(C_e) \setminus \{b_e\}$,
the unique $s$--$w$ path $Q_w$ such that $e \in E(Q_w) \subseteq T \cup \{e\}$
satisfies the following properties.
\begin{itemize}
  \setlength{\itemsep}{.5mm}
\item[$(1)$]
  For any vertex $z \in V(P_{b_e}) \subsetneq V(Q_w)$, we have $\psi_G(Q_w[z, w]) \neq \psi_T(z)^{-1} \cdot \psi_T(w) = \psi_G(P_w[z, w])$, i.e., $Q_w[z, w]$ is unorthodox.
\item[$(2)$]
  For any vertex $z \in V(P_{b_e}) \subsetneq V(Q_w)$ and any unorthodox $z$--$w$ path $R$ in $G$, we have $\ell(R) \geq \ell(Q_w[z, w])$.
\item[$(3)$]
  $Q_w$ is a shortest unorthodox $s$--$w$ path in $G$, i.e., minimizes $\ell(Q_w)$ subject to $\psi_G(Q_w) \neq \psi_T(w)$.
\end{itemize}
\end{lemma}

\begin{proof}
$(1)$~
By definition, $Q_w * \overline{P_w}$ coincides with either $W_e = P_u * (u, e, v) * \overline{P_v}$ itself or the reversed walk $\overline{W_e}$. Since $Q_w[s, z] = P_z$ and $\psi_G(W_e) = \psi_T(u) \cdot \psi_G(e, uv) \cdot \psi_T(v)^{-1} \neq \id$, we have
\begin{align*}
\psi_G(Q_w[z, w]) = \psi_G(P_z)^{-1} \cdot \psi_G(W_e)^{\pm 1} \cdot \psi_G(P_w)
\neq \psi_G(P_z)^{-1} \cdot \id \cdot \psi_G(P_w)
= \psi_G(P_w[z, w]).
\end{align*}

\noindent
$(2)$~
By Lemma~\ref{lem:tree}(2), there exists an inconsistent edge $f = \{x, y\} \in E(R) \setminus T$,
and we can write $R = R[z, x] * (x, f, y) * R[y, w]$ (possibly, $x = z$ or $y = w$).
By Lemma~\ref{lem:tree}(1), we have $\ell(R[z, x]) \geq \dist_T(x) - \dist_T(z)$ and $\ell(R[y, w]) \geq \dist_T(y) - \dist_T(w)$.
Since $C_e$ is a lowest blossom, we have
\begin{align*}
  \ell(R) &= \ell(R[z, x]) + \ell(f) + \ell(R[y, w])\\
  &\geq \dist_T(x) + \dist_T(y) + \ell(f) - \dist_T(w) - \dist_T(z)\\
  &\geq \dist_T(u) + \dist_T(v) + \ell(e) - \dist_T(w) - \dist_T(z)\\
  &= \ell(W_e) - \ell(P_w) - \ell(P_z)\\ &= \ell(Q_w) - \ell(P_z) = \ell(Q_w[z, w]).
\end{align*}

\noindent
$(3)$~
Just combine (1) and (2) by choosing $z = s$.
\end{proof}

Fix a lowest blossom $C$ with base $b$,
and let $Q_w$ denote the shortest unorthodox $s$--$w$ path defined in Lemma~\ref{lem:second} for each $w \in V(C) \setminus \{b\}$.
If $t \in V(C) \setminus \{b\}$, then $Q_t$ is a desired $s$--$t$ path, and we are done.
Otherwise, we shrink $C$ into $b$,
and recursively find a shortest unorthodox $s$--$t$ path in the resulting graph,
which can be expanded into a shortest unorthodox $s$--$t$ path in $G$.

The shrinking and expanding operations are formally described in Definitions~\ref{def:shrink} and \ref{def:expand}.
Intuitively, for each vertex $w \in V(C) \setminus \{b\}$ (which is removed by shrinking), we care about the only two $b$--$w$ paths along $C$, say $R_{b, w}^1 = P_w[b, w]$ and $R_{b, w}^2 = Q_w[b, w]$, and for each edge $f = \{w, x\} \in \delta_G(w)$ with $x \not\in V(C)$, we create two new edges $\tf_i = \{b, x\}$ $(i = 1, 2)$ corresponding to two $b$--$x$ paths $R_{b, w}^i * (w, f, x)$ in $G$ (see Fig.~\ref{fig:shrink}).

\begin{definition}[Shrinking]\label{def:shrink}
For an $s$-SPT $T$ of $(G, \ell)$ and a lowest blossom $C$ with base $b$,
we say that a $\Gamma$-labeled graph $\tG = (\tV, \tE)$ with edge lengths $\tell \in \RR_{\geq 0}^\tE$ and a spanning tree $\tT \subseteq \tE$ (cf.~Lemma~\ref{lem:shrink}(3)) is obtained by \emph{shrinking $C$ into $b$}
(and denote $\tilde{\bullet}$ by $\bullet[C \to b]$ for $\bullet \in \{G, \ell, T\}$) if it is defined as follows (see Fig.~\ref{fig:shrink}).
\begin{itemize}
  \setlength{\itemsep}{.5mm}
\item $\tV \coloneqq V \setminus (V(C) \setminus \{b\})$.
\item $\tE \coloneqq (E \setminus E_{C, b}) \cup \tF_{C, b}$ and $\tilde{T} \coloneqq (T \setminus E_{C, b}) \cup \{\, \tf_1 \mid f \in T \cap (E_{C, b} \setminus E(C)) \,\}$, where
\begin{align*}
  E_{C, b} &\coloneqq \{\, e \in E \mid e \cap (V(C) \setminus \{b\}) \neq \emptyset \,\},\\
  \tF_{C, b} &\coloneqq \{\, \tf_i = \{b, x\} \mid f = \{w, x\} \in E_{C, b},~w \in V(C) \setminus \{b\},~x \not\in V(C),~i \in \{1, 2\} \,\}.
\end{align*}
\item The labels and lengths are defined as follows.\vspace{-1mm}
\begin{itemize}
\setlength{\itemsep}{1mm}
\item For each $e = \{u, v\} \in E \setminus E_{C, b}$, define $\psi_\tG(e, uv) \coloneqq \psi_G(e, uv)$ and $\tell(e) \coloneqq \ell(e)$.
\item For $\tf_i = \{b, x\} \in \tF_{C, b}$ with $f = \{w, x\}$,
  define $\psi_\tG(\tf_i, bx) = \psi_G(R_{b, w}^i) \cdot \psi_G(f, wx)$ and $\tell(\tf_i) \coloneqq \ell(R_{b, w}^i) + \ell(f)$,
  where $R_{b, w}^1 \coloneqq P_w[b, w]$ and $R_{b, w}^2 \coloneqq Q_w[b, w]$.
\end{itemize}
\end{itemize}
\end{definition}

\begin{figure}[tbp]
 \begin{center}
  \includegraphics[scale=0.8]{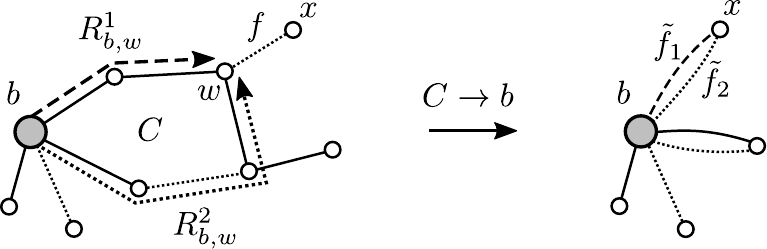}
 \end{center}\vspace{-3mm}
 \caption{Shrinking $C$ into $b$, where solid lines represent edges in $T$ and in $\tT$, respectively.}%\vspace{-1mm}
 \label{fig:shrink}
\end{figure}

\begin{definition}[Expanding]\label{def:expand}
Under the same setting as Definition~\ref{def:shrink}, let $v \in \tV$.
For an \mbox{$s$--$v$} path $\tR_v$ in $\tG$, we say that an $s$--$v$ walk $R_v$ in $G$ is obtained by \emph{expanding $b$ into $C$}
(and denote $R_v$ by $\tR_v[b \to C]$) if it is obtained from $\tR_v$ as follows
(see also Fig.~\ref{fig:shrink}):
for each $\tf_i = \{b, x\} \in E(\tR_v) \cap \tF_{C, b}$ with $f = \{w, x\} \in E_{C, b}$ and $i \in \{1, 2\}$,
replace the subpath $(b, \tf_i, x)$ or $(x, \tf_i, b)$ appearing in $\tR_v$ with the corresponding path $R_{b, w}^i * (w, f, x)$ or $(x, f, w) * \overline{R_{b, w}^i}$ in $G$, respectively.
\end{definition}

We observe basic properties on these operations.

\begin{lemma}\label{lem:shrink}
The following properties hold in Definitions~$\ref{def:shrink}$ and $\ref{def:expand}$.
\begin{itemize}
  \setlength{\itemsep}{.5mm}
\item[$(1)$]
  $\ell(R_v) = \tell(\tR_v)$ and $\psi_G(R_v) = \psi_\tG(\tR_v)$.
\item[$(2)$]
  If $|E(\tR_v) \cap \tF_{C, b}| \leq 1$, then $R_v$ is an $s$--$v$ path.
  Otherwise, $|E(\tR_v) \cap \tF_{C, b}| = 2$ holds and each vertex in $\tV = V \setminus (V(C) \setminus \{b\})$ appears in $R_v$ at most once. 
\item[$(3)$]
$\tT$ is an $s$-SPT of $(\tG, \tell)$ such that, for every $v \in \tV$, we have $\dist_\tT(v) = \dist_T(v)$ and $\psi_\tT(v) = \psi_T(v)$,
and a unique $s$--$v$ path $\tP_v$ in $\tT$ is expanded into $P_v$, i.e., $\tP_v[b \to C] = P_v$.
\end{itemize}
\end{lemma}

\begin{proof}
The properties $(1)$ and $(2)$ immediately follow from the definitions (see $\psi_\tG$ and $\tell$, and note that $\tR_v$ is chosen as an $s$--$v$ path in $\tG$).
In what follows, we prove $(3)$.

Since $T$ is a spanning tree of $G$ and $C$ is a cycle with $|E(C) \setminus T| = 1$, we see that $\tT$ is indeed a spanning tree of $\tG = G[C \to b]$, in which $|E(\tP_v) \cap \tF_{C, b}| \leq 1$ and $\tP_v[b \to C] = P_v$ for each $v \in \tV$ by definition.
By the property (1), we have $\dist_\tT(v) = \dist_T(v)$ and $\psi_\tT(v) = \psi_T(v)$.
Moreover, for any $s$--$v$ path $\tR_v$ in $\tG$, there exists a corresponding $s$--$v$ walk $R_v = \tR_v[b \to C]$ in $G$ such that $\ell(R_v) = \tell(\tR_v)$, which implies the existence of a corresponding $s$--$v$ path $R'_v$ in $G$ with $\ell(R'_v) \leq \ell(R_v) = \tell(\tR_v)$ by the nonnegativity of $\ell$.
Since $P_v$ is shortest in $G$, so is $\tP_v$ in $\tG$.
\end{proof}

Whereas we want to obtain a shortest unorthodox path in the original graph $G$, the expanding operation does not necessarily yield a path by definition.\footnote{For the concept of expanding, two different definitions that always yield a path were employed in the submitted and final versions of SODA 2020, but both of them turned out to fail sometimes. Thus we have modified them in the present way by separating expansion and simplification. See also Remark after Definition~\ref{def:simplify}.}
In order to complete our recursive algorithm (and also to utilize it in the correctness proof), we define one more operation to obtain an unorthodox $s$--$v$ path in $G$ by simplifying an unorthodox $s$--$v$ walk in $G$ in which each vertex in $V \setminus (V(C) \setminus \{b\})$ appears at most once (cf.~Lemma~\ref{lem:shrink}(2)).

\begin{definition}[Simplification]\label{def:simplify}
Under the same setting as Definitions~\ref{def:shrink} and \ref{def:expand}, let $R_t$ be an unorthodox $s$--$t$ walk in $G$ in which each vertex in $V \setminus (V(C) \setminus \{b\})$ appears at most once.
We then define the \emph{simplification} of $R_t$ as an unorthodox $s$--$t$ path $R'_t$ with $E(R'_t) \subseteq E(R_t) \cup E(C) \cup E(P_b)$ as follows (see Figs.~\ref{fig:case_1}--\ref{fig:case_2_2}).
If $V(R_t) \cap (V(C) \setminus \{b\}) = \emptyset$, then let $R'_t \coloneqq R_t$, which is indeed a path by the assumption.
Otherwise, let $w$ be the last vertex on $R_t$ intersecting the blossom $C$ or its stem $P_b$, so that $R_t[w, t]$ is a path with $V(R_t[w, t]) \cap (V(C) \cup V(P_b)) = \{w\}$.
\begin{description}
\setlength{\itemsep}{0mm}
  \item[Case~1.] When $w \in V(C) \setminus \{b\}$ (see Fig.~\ref{fig:case_1}). \vspace{-.5mm}
  \begin{description}
  \setlength{\itemsep}{.5mm}
    \item[Case~1.1.] If $\psi_G(P_w) \cdot \psi_G(R_t[w, t]) \neq \psi_T(t)$, then let $R'_t \coloneqq P_w * R_t[w, t]$.
    \item[Case~1.2.] Otherwise, let $R'_t \coloneqq Q_w * R_t[w, t]$.
  \end{description}
 
\begin{figure}[tb]
 \begin{center}
  \includegraphics[scale=0.8]{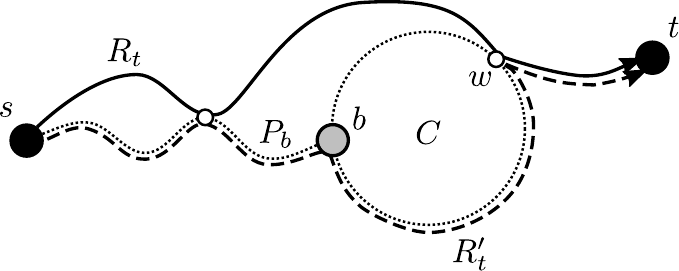}
 \end{center}\vspace{-3mm}
 \caption{General picture of Case~1 (when $w \in V(C) \setminus \{b\}$), where the solid line represents $R_t$, the dashed one is $R'_t$ with $R'_t[s, w] = P_w$ or $Q_w$, and the dotted ones are $P_b$ and $C$.}%\vspace{-1mm}
 \label{fig:case_1}
\end{figure}
 
  \item[Case~2.] When $w \in V(P_b)$ (see Figs.~\ref{fig:case_2_1} and \ref{fig:case_2_2}).
  
  Let $u$ and $v$ be the first and last vertices, respectively, on $R_t$ intersecting $V(C) \setminus \{b\}$, so that $R_t[s, u]$ and $R_t[v, t]$ are paths with $V(R_t[s, u]) \cap (V(C) \setminus \{b\}) = \{u\}$ and $V(R_t[v, t]) \cap (V(C) \setminus \{b\}) = \{v\}$ (possibly $u = v$).
  %Since $w \in V(P_b)$ appears in $R_t$ later than $u$ and $v$, we can write $R_t = R_t[s, u] * R_t[u, v] * R_t[v, w] * R_t[w, t]$.
  Let $y_1 \in V(R_t[s, u]) \cap V(P_b)$ and $y_2 \in V(R_t[v, t]) \cap V(P_b)$ be respectively the farthest vertices on $P_b$ from $s$, so that $V(R_t[s, u]) \cap V(P_b[y_1, b]) = \{y_1\}$ and $V(R_t[v, t]) \cap V(P_b[y_2, b]) = \{y_2\}$ (possibly, $y_1 = s$ or $y_2 = w$).
  As $y_1, y_2 \in V \setminus (V(C) \setminus \{b\})$, we have $y_1 \neq y_2$ by the assumption on $R_t$, and let $y$ be the farther one, i.e., $y \coloneqq y_1$ if $y_1 \in V(P_b[y_2, b])$ and $y \coloneqq y_2$ otherwise (then $y_2 \in V(P_b[y_1, b])$).%\vspace{-1mm}
  \begin{description}
    \setlength{\itemsep}{.5mm}
    \item[Case~2.1.] When $y = y_1$ (see Fig.~\ref{fig:case_2_1}).
      \begin{description}
      \setlength{\itemsep}{1mm}
      \item[Case~2.1.1] If $\psi_G(R_t[s, y]) \cdot \psi_G(P_v[y, v]) \cdot \psi_G(R_t[v, t]) \neq \psi_T(t)$, then let $R'_t \coloneqq R_t[s, y] * P_v[y, v] * R_t[v, t]$.
      \item[Case~2.1.2.] Otherwise, let $R'_t \coloneqq R_t[s, y] * Q_v[y, v] * R_t[v, t]$.
    \end{description}
    \item[Case~2.2.] When $y = y_2$ (see Fig.~\ref{fig:case_2_2}).
    \begin{description}
    \setlength{\itemsep}{1mm}
      \item[Case~2.2.1.] If $\psi_G(R_t[s, u]) \cdot \psi_G(\overline{P_u}[u, y]) \cdot \psi_G(R_t[y, t]) \neq \psi_T(t)$, then let $R'_t \coloneqq R_t[s, u] * \overline{P_u}[u, y] * R_t[y, t]$.
      \item[Case~2.2.2.] Otherwise, let $R'_t \coloneqq R_t[s, u] * \overline{Q_u}[u, y] * R_t[y, t]$.
    \end{description}
  \end{description}
\end{description}
\end{definition}

\begin{remark}
We remark why we have to employ this complex definition instead of a rather simple approach (which was actually employed in the first manuscript submitted to SODA 2020) as follows.
Let $\tilde{R}_t$ be an unorthodox $s$--$t$ path in $\tilde{G}$ that is expanded to $R_t$ in the sense of Definition~\ref{def:expand}.
By Lemma~\ref{lem:shrink}, if $|E(\tilde{R}_t) \cap \tilde{F}_{C, b}| \leq 1$, then $R_t$ is indeed a desired path (which is unorthodox and is not longer than $\tilde{R}_t$), and otherwise $|E(\tilde{R}_t) \cap \tilde{F}_{C, b}| = 2$.
If the two corresponding edges in $E_{C, b}$ (cf.~Definition~\ref{def:shrink}) have different endpoints in $C$, say $u$ and $w$, then we have two possible choices of $u$--$w$ subpaths in $C$ whose labels are different (as $C$ is non-zero), which enable us to choose one of them even if $R_t[u, w]$ is not a path but a walk including $C$.
However, it may happen that the two edges have the same endpoints in $C$ and $R_t$ includes $C$.
In this case, to obtain an $s$--$t$ path in $G$, we must remove $C$ from $R_t$, which changes the label into an undesirable one (as $C$ is non-zero).
\end{remark}

%We show that the simplification is indeed the desired operation.

\begin{lemma}\label{lem:simplify}
$R'_t$ defined in Definition~$\ref{def:simplify}$ is an $s$--$t$ path in $G$ with $\psi_G(R'_t) \neq \psi_T(t)$ and $\ell(R'_t) \leq \ell(R_t)$.
\end{lemma}

\begin{figure}[tb]
   \begin{center}
    \includegraphics[scale=0.8]{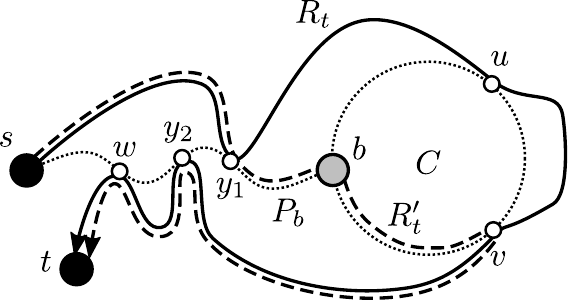}
   \end{center}\vspace{-3mm}
   \caption{General picture of Case~2.1 (when $y = y_1$), where the solid line represents $R_t$, the dashed one is $R'_t$ with $R'_t[y, v] = P_v[y, v]$ or $Q_v[y, v]$, and the dotted ones are $P_b$ and $C$.}\vspace{7mm}
   \label{fig:case_2_1}
\end{figure}

\begin{figure}[tb]
 \begin{center}
  \includegraphics[scale=0.8]{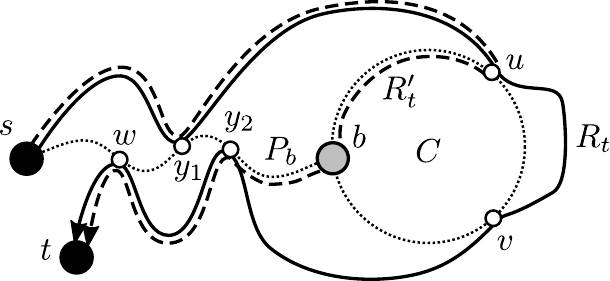}
 \end{center}\vspace{-3mm}
 \caption{General picture of Case~2.2 (when $y = y_2$), where the solid line represents $R_t$, the dashed one is $R'_t$ with $R'_t[u, y] = \overline{P_u}[u, y]$ or $\overline{Q_u}[u, y]$, and the dotted ones are $P_b$ and $C$.}%\vspace{-1mm}
 \label{fig:case_2_2}
\end{figure}

\begin{proof}
The case when $R'_t = R_t$ is trivial.
We confirm the other six cases separately as follows.

\medskip\noindent
{{\bf Case~1.}~~When $w \in V(C) \setminus \{b\}$ (see Fig.~\ref{fig:case_1}).}

\medskip
\noindent
{\bf Case~1.1.}~
In this case, we have $\psi_G(R'_t)= \psi_G(P_w) \cdot \psi_G(R_t[w, t]) \neq \psi_T(t)$.
As $P_w$ is a shortest $s$--$w$ path in $G$, we see
\begin{align*}
  %\psi_G(R'_t) &= \psi_G(P_w) \cdot \psi_G(R_t[w, t]) = \psi_T(w) \cdot \psi_G(R_t[w, t]) \neq \psi_T(t),\\[1mm]
  \ell(R'_t) &= \ell(P_w) + \ell(R_t[w, t])\leq \ell(R_t[s, w]) + \ell(R_t[w, t]) = \ell(R_t).
\end{align*}

\smallskip\noindent
{\bf Case~1.2.}~
In this case, we have $\psi_G(P_w) \cdot \psi_G(R_t[w, t]) = \psi_T(t)$, which implies $R_t[s, w]$ is unorthodox as so is $R_t$.
By Lemma~\ref{lem:second}, we see
\begin{align*}
  \psi_G(R'_t) &= \psi_G(Q_w) \cdot \psi_G(R_t[w, t])= \psi_G(Q_w) \cdot \psi_T(w)^{-1} \cdot \psi_T(t) \neq \psi_T(t),\\[1mm]
  \ell(R'_t) &= \ell(Q_w) + \ell(R_t[w, t]) \leq \ell(R_t[s, w]) + \ell(R_t[w, t]) = \ell(R_t).
\end{align*}

\smallskip\noindent
{{\bf Case~2.1.}~~When $w \in V(P_b)$ and $y = y_1$ (see Fig.~\ref{fig:case_2_1}).}

\smallskip
In this case, $V(R_t[s, y]) \cap (V(C) \cup V(P_b[y, b])) = \{y\}$ and $V(R_t[v, t]) \cap (V(C) \cup V(P_b[y, b])) = \{v\}$.
Hence, $R'_t$ is indeed an $s$--$t$ path because $R'_t[y, v]$ is a subpath of $P_b[y, b] * C$ or $P_b[y, b] * \overline{C}$.

\medskip\noindent
{\bf Case~2.1.1.}~
In this case, we have $\psi_G(R'_t) = \psi_G(R_t[s, y]) \cdot \psi_G(P_v[y, v]) \cdot \psi_G(R_t[v, t]) \neq \psi_T(t)$.
As $P_v = P_y * P_v[y, v]$, by Lemma~\ref{lem:tree}(1), we see
\begin{align*}
  %\psi_G(R'_t) &= \psi_G(R_t[s, y]) \cdot \psi_T(y)^{-1} \cdot \psi_T(v) \cdot \psi_G(R_t[v, t]) \neq \psi_T(t),\\[1mm]
  \ell(R'_t) &= \ell(R_t[s, y]) + \left(\dist_T(v) - \dist_T(y)\right) + \ell(R_t[v, t])\\ &\leq \ell(R_t[s, y]) + \ell(R_t[y, v]) + \ell(R_t[v, t]) = \ell(R_t).
\end{align*}

\noindent
{\bf Case~2.1.2.}~
In this case, we have $\psi_G(R_t[s, y]) \cdot \psi_G(P_v[y, v]) \cdot \psi_G(R_t[v, t]) = \psi_T(t)$, which implies $R_t[y, v]$ is unorthodox as so is $R_t$.
As $y \in V(P_b)$ and $v \in V(C) \setminus \{b\}$, by Lemma~\ref{lem:second}, we see
\begin{align*}
  \psi_G(R'_t) &= \psi_G(R_t[s, y]) \cdot \psi_G(Q_v[y, v]) \cdot \psi_G(R_t[v, t])\\
   &\neq \psi_G(R_t[s, y]) \cdot \psi_G(P_v[y, v]) \cdot \psi_G(R_t[v, t]) = \psi_T(t),\\[1mm]
  \ell(R'_t) &= \ell(R_t[s, y]) + \ell(Q_v[y, v]) + \ell(R_t[v, t])\\
  &\leq \ell(R_t[s, y]) + \ell(R_t[y, v]) + \ell(R_t[v, t]) = \ell(R_t).
\end{align*}

\smallskip
\noindent
{{\bf Case~2.2.}~~When $w \in V(P_b)$ and $y = y_2$ (see Fig.~\ref{fig:case_2_2}).}

\smallskip
In this case, $V(R_t[s, u]) \cap (V(C) \cup V(P_b[y, b])) = \{u\}$ and $V(R_t[y, t]) \cap (V(C) \cup V(P_b[y, b])) = \{y\}$.
Hence, $R'_t$ is indeed an $s$--$t$ path because $R'_t[u, y]$ is a subpath of $C * \overline{P_b}[b, y]$ or $\overline{C} * \overline{P_b}[b, y]$.

\medskip\noindent
{\bf Case~2.2.1.}~
In this case, we have $\psi_G(R'_t) = \psi_G(R_t[s, u]) \cdot \psi_G(\overline{P_u}[u, y]) \cdot \psi_G(R_t[y, t]) \neq \psi_T(t)$.
As $\overline{P_u} = \overline{P_u}[u, y] * \overline{P_y}$, by Lemma~\ref{lem:tree}(1), we see
\begin{align*}
  \ell(R'_t) &= \ell(R_t[s, u]) + \left(\dist_T(u) - \dist_T(y)\right) + \ell(R_t[y, t])\\
  &\leq \ell(R_t[s, u]) + \ell(R_t[u, y]) + \ell(R_t[y, t])= \ell(R_t).
\end{align*}

\noindent
{\bf Case~2.2.2.}~
In this case, we have $\psi_G(R_t[s, u]) \cdot \psi_G(\overline{P_u}[u, y]) \cdot \psi_G(R_t[y, t]) = \psi_T(t)$, which implies $\overline{R_t}[y, u]$ is unorthodox as so it $R_t$.
As $y \in V(P_b)$ and $u \in V(C) \setminus \{b\}$, by Lemma~\ref{lem:second}, we see
\begin{align*}
  \psi_G(R'_t) &= \psi_G(R_t[s, u]) \cdot \psi_G(Q_u[y, u])^{-1} \cdot \psi_G(R_t[y, t])\\
   &\neq \psi_G(R_t[s, u]) \cdot \psi_G(P_u[y, u])^{-1} \cdot \psi_G(R_t[y, t]) = \psi_T(t),\\[1mm]
  \ell(R'_t) &= \ell(R_t[s, u]) + \ell(Q_u[y, u]) + \ell(R_t[y, t])\\
  &\leq \ell(R_t[s, u]) + \ell(\overline{R_t}[y, u]) + \ell(R_t[y, t]) = \ell(R_t). \qedhere
\end{align*}
\end{proof}

Lemma~\ref{lem:simplify} implies the other of the two key properties of lowest blossoms as follows.
One can obtain a shortest unorthodox $s$--$t$ path in $G$ by doing so recursively in the graph $\tG$ obtained by shrinking a lowest blossom into its base, which completes the correctness of our recursive strategy to solve the shortest unorthodox path problem.

\begin{lemma}\label{lem:C_e}
For an $s$-SPT $T$ of $(G, \ell)$ and a lowest blossom $C$ with base $b$,
let $\tG \coloneqq G[{C \to b}]$, $\tell \coloneqq \ell[{C \to b}]$, and $\tT \coloneqq T[C \to b]$.
If $t \in \tV$, then the following properties hold.
\begin{itemize}
  \setlength{\itemsep}{.5mm}
\item[$(1)$]
If $G$ has an unorthodox $s$--$t$ path (with respect to $T$), then so does $\tG$ (with respect to $\tT$).
\item[$(2)$]
For any shortest unorthodox $s$--$t$ path $\tQ_t$ in $\tG$,
the simplification of $\tQ_t[b \to C]$ is a shortest unorthodox $s$--$t$ path in $G$.
\end{itemize}
\end{lemma}

\begin{proof}
Fix an unorthodox $s$--$t$ path $R_t$ in $G$, and let $R'_t$ be its simplification (possibly, $R'_t = R_t$).
We then have $\psi_G(R'_t) \neq \psi_T(t)$ and $\ell(R'_t) \leq \ell(R_t)$ by Lemma~\ref{lem:simplify}.
Moreover, in any case of Definition~\ref{def:simplify}, there exists a corresponding $s$--$t$ path $\tR'_t$ in $\tG$ with $|E(\tR'_t) \cap \tF_{C, b}| \leq 1$ and $\tR'_t[b \to C] = R'_t$.
By Lemma~\ref{lem:shrink}, we see the property (1) as $\psi_\tG(\tR'_t) = \psi_G(R'_t) \neq \psi_T(t) = \psi_\tT(t)$ and the property (2) as
\begin{align*}
  \ell(Q_t) \leq \ell(\tQ_t[b \to C]) = \tell(\tQ_t) \leq \tell(\tR'_t) = \ell(R'_t) \leq \ell(R_t),
\end{align*}
where $Q_t$ denotes the simplification of $\tQ_t[b \to C]$.
\end{proof}

\section{An $\mathrm{O}(nm)$-Time Algorithm with Blossom Shrinking}\label{sec:algorithm}
In this section, we present an $\mathrm{O}(nm)$-time recursive algorithm for the shortest unorthodox path problem.
In addition, as a byproduct, we show an $\mathrm{O}(n(m + n \log n))$-time algorithm for finding a shortest non-zero cycle, which completes the proof of Theorem~\ref{thm:cycle}.

\subsection{Algorithm Description}\label{sec:description}
Recall that the input graph $G$ is assumed to be connected (by extracting the connected component that contains $s$ and $t$ in advance if necessary).
After obtaining a shortest-path tree $T$ of $(G, \ell)$ rooted at $s$ by {\sc Dijkstra}$[G, \ell, s]$, we can use the following algorithm to solve the shortest unorthodox path problem, whose correctness has already been confirmed in Section~\ref{sec:CUC}.

\begin{algorithm*}[{{\sc SUP}$[G, \ell, s, t, T]$}]
\begin{description}
  \setlength{\itemsep}{0mm}
\item[]

\item[Input:]
  A connected $\Gamma$-labeled graph $G = (V, E)$, edge lengths $\ell \in \RR_{\geq 0}^E$, distinct vertices $s, t \in V$,
  and an $s$-SPT $T$ of $(G, \ell)$ with $\dist_T$ and $\psi_T$.

\item[Output:]
  A shortest unorthodox $s$--$t$ path in $G$ if any, or the message ``NO'' otherwise.
  
\item[Step 1.]\vspace{1mm}
  If every edge $e = \{u, v\} \in E \setminus T$ is consistent with $T$ (i.e., $\psi_T(u) \cdot \psi_G(e, uv) = \psi_T(v)$), then halt with the message ``NO'' (cf.~Lemma~\ref{lem:tree}(2)).

\item[Step 2.]
  Otherwise, pick an inconsistent edge $e = \{u, v\}$ that minimizes $\dist_T(u) + \dist_T(v) + \ell(e)$,
  so that $C_e = P_u[b_e, u] * (u, e, v) * \overline{P_v}[v, b_e]$ is a lowest blossom with base $b_e$ (cf.~Definition~\ref{def:canonical}).
  
\item[Step 3.]
  If $t \in V(C_e) \setminus \{b_e\}$, then halt with returning the $s$--$t$ path $Q_t$ in $G$ obtained as follows (cf.~Lemma~\ref{lem:second}(3)):
  $Q_t \coloneqq P_u * (u, e, v) * \overline{P_v}[v, t]$ if $t \in V(P_v) \setminus V(P_u)$,
  and $Q_t \coloneqq P_v * (v, e, u) * \overline{P_u}[u, t]$ otherwise (then, $t \in V(P_u) \setminus V(P_v)$).
  
\item[Step 4.]
  Otherwise, construct a $\Gamma$-labeled graph $\tG = (\tV, \tE)$ with edge lengths $\tell \in \RR_{\geq 0}^\tE$ and an $s$-SPT $\tT \subseteq \tE$ of $(\tG, \tell)$ by shrinking $C_e$ into $b_e$, i.e., set $\tG \leftarrow G[C_e \to b_e]$, $\tell \leftarrow \ell[C_e \to b_e]$, and $\tT \leftarrow T[C_e \to b_e]$ (cf.~Definition~\ref{def:shrink}), where $\dist_\tT$ and $\psi_\tT$ are obtained just by restricting $\dist_T$ and $\psi_T$, respectively, to $\tV$ (cf.~Lemma~\ref{lem:shrink}(3)).
  In order to keep the resulting graph not having too many edges, we remove redundant edges if arise (see Section~\ref{sec:time} for the detail).
  
\item[Step 5.]
  Perform {\sc SUP}$[\tG, \tell, s, t, \tT]$ recursively, and halt with the following output.\vspace{-1.5mm}
  \begin{itemize}
  \setlength{\itemsep}{.5mm}
  \item
    If the output of {\sc SUP}$[\tG, \tell, s, t, \tT]$ is a shortest unorthodox $s$--$t$ path $\tQ_t$ in $\tG$, then return the simplification of the expanded $s$--$t$ walk $\tQ_t[b_e \to C_e]$ in $G$ (cf.~Definitions~\ref{def:expand} and \ref{def:simplify} and Lemma~\ref{lem:C_e}(2)).
  \item
    Otherwise, return the message ``NO'' (cf.~Lemma~\ref{lem:C_e}(1)).
  \end{itemize}
\end{description}
\end{algorithm*}

As remarked in the introduction just before Theorem~\ref{thm:cycle}, we also find a shortest non-zero cycle by computing one lowest blossom $C_r$ with respect to one $r$-SPT $T_r$ of $(G, \ell)$ for each vertex $r \in V$, whose correctness is seen as follows.

Fix a vertex $r \in V$, an $r$-SPT $T_r$ of $(G, \ell)$, and a lowest blossom $C_r$.
Let $e = \{u, v\} \in E(C_r) \setminus T_r$ be the corresponding inconsistent edge,
and define $W_e \coloneqq P_u * (u, e, v) * P_v$ as a closed walk in $G$ with end vertex $r$ that is included in $T_r \cup \{e\}$ (cf.~Fig.~\ref{fig:C_e}).
Then, by Lemma~\ref{lem:lb}, $W_e$ is a shortest non-zero closed walk with end vertex $r$,
which implies that $C_r$ is a shortest non-zero cycle in $G$ if $r$ is on some shortest non-zero cycle in $G$
by the nonnegativity of edge length.

The algorithm for finding a shortest non-zero cycle is formally described as follows.
Recall that Dijkstra's algorithm can be implemented in ${\rm O}(m + n \log n)$ time with the aid of Fibonacci heaps.
Thus, the total computational time is ${\rm O}(n(m + n \log n))$, which completes the proof of Theorem~\ref{thm:cycle}.

\begin{algorithm*}[{{\sc ShortestNon-zeroCycle}$[G, \ell]$}]
\begin{description}
  \setlength{\itemsep}{0mm}
\item[]
\item[Input:]
  A connected $\Gamma$-labeled graph $G = (V, E)$ and edge lengths $\ell \in \RR_{\geq 0}^E$.

\item[Output:]
  A shortest non-zero cycle in $G$ if any, or the message ``NO'' otherwise.
  
\item[Step 1.]\vspace{1mm}
  For each vertex $r \in V$, compute an $r$-SPT $T_r$ of $(G, \ell)$ by {\sc Dijkstra}$[G, \ell, r]$,
  and then find a lowest blossom $C_r$ with respect to $T_r$ (cf.~Step~2 of {\sc SUP}).
  
\item[Step 2.]
  Return the shortest cycle among $C_r$ $(r \in V)$ if at least one has been found, and ``NO'' otherwise.
\end{description}
\end{algorithm*}

\subsection{Computational Time Analysis}\label{sec:time}
In this section, we show that the computational time of {\sc SUP}$[G, \cdots]$ is bounded by ${\rm O}(nm)$.

First of all, since $G$ is connected, we have $n = {\rm O}(m)$.
In addition, although $G$ may have arbitrarily many parallel edges in general, we can reduce them to at most two between each pair of two vertices so that $m = {\rm O}(n^2)$ in advance, because more than two parallel edges between the same pair are \emph{redundant} for our problems as follows.
If there are two parallel edges with the same label, then one that is not longer than the other is enough.
Moreover, if there are three parallel edges with distinct labels, then any non-zero (or unorthodox) $s$--$t$ path traversing a longest one among them can be transformed into another non-zero (or unorthodox, respectively) $s$--$t$ path not longer than the original path by replacing the edge with at least one of the other two.
Starting with the edgeless graph with vertex set $V$, one can construct such a necessary subgraph of $G$ in ${\rm O}(m)$ time by sequentially adding the edges in an arbitrary order and deciding which one (or two) should be retained whenever we encounter a new edge that is parallel with some retained edge.%\footnote{One can check whether an edge already exists between the two vertices in $\mathrm{O}(1)$ time in average by using a hash table, or $\mathrm{O}(\log n)$ time in the worst case by using a binary search tree (see \cite{CLRS2009} for the basics on the data structures).}

We now proceed to the main part.
Since the input graph of a recursive call after shrinking a lowest blossom into the base has strictly fewer vertices, the number of recursive calls is at most $n$.
In addition, as shown later in Claim~\ref{cl:edges}, the number of edges in any recursion is bounded by ${\rm O}(m)$
if we reduce redundant parallel edges whenever some appears as a new edge not contained in the new shortest-path tree
(then, each pair of two vertices has at most two parallel edges).
Hence, it suffices to show that {\sc SUP}$[G, \cdots]$ can be implemented in ${\rm O}(m)$ time except for the recursive call.

It is easy to confirm that a na\"ive implementation of Steps~1--2 requires ${\rm O}(m)$ time.
In addition,
Steps~3 and~5 (except for the recursive call {\sc SUP}$[\tG, \cdots]$) are done in ${\rm O}(n)$ time,
because the relevant edges in $G$ are included in $T \cup \{e\} \cup (E(\tQ_t) \cap E)$, whose size is ${\rm O}(n)$
as $T$ is a spanning tree of $G$ and $\tQ_t$ is a path in $\tG$.

Let us focus on Step~4, i.e., shrinking a lowest blossom $C_e$ into the base $b_e$.
We use the same notation as Definitions~\ref{def:canonical} and \ref{def:shrink} by setting $C = C_e$ and $b = b_e$.
Since the number of removed edges and new edges is ${\rm O}(m)$ by definition,
this can be implemented in ${\rm O}(m)$ time in total 
by computing the labels and lengths of new edges in constant time per edge as follows.
Consider new edges $\tf_i = \{b, x\}$ $(i = 1, 2)$ corresponding to a removed edge $f = \{w, x\} \in E_{C, b}$.
Since $R_{b, w}^1 = P_w[b, w]$ and $R_{b, w}^2 = Q_w[b, w]$,
we have $\psi_G(R_{b, w}^1) = \psi_T(b)^{-1} \cdot \psi_T(w)$, $\ell(R_{b, w}^1) = \dist_T(w) - \dist_T(b)$,
$\psi_G(R_{b, w}^2) = \psi_T(b)^{-1} \cdot \psi_G(Q_w)$, and $\ell(R_{b, w}^2) = \ell(Q_w) - \dist_T(b)$, where
\begin{align*}
\psi_G(Q_w) &= \begin{cases}
  \psi_G(W_e) \cdot \psi_T(w) & (w \in V(P_v)),\\
  \psi_G(W_e)^{-1} \cdot \psi_T(w) & (w \in V(P_u)),
\end{cases}\\[1mm]
\ell(Q_w) &= \ell(W_e) - \dist_T(w),\\[1mm]
\psi_G(W_e) &= \psi_T(u) \cdot \psi_G(e, uv) \cdot \psi_T(v)^{-1},\\[1mm]
\ell(W_e) &= \dist_T(u) + \dist_T(v) + \ell(e).
\end{align*}
Using these equations, we can compute $\psi_\tG(\tf_i, bx)$ and $\tell(\tf_i)$ by a constant number of elementary operations.
In addition, if some $\tf_i \in \tF_{C, b} \setminus \tT$ is redundant,
then we immediately remove it\footnote{Even if the graph is represented by an adjacency list, this can be done in $\mathrm{O}(1)$ time on average by using a hash table in practice, or in $\mathrm{O}(n^2 + m) = \mathrm{O}(nm)$ time in total by constructing an adjacency matrix at the beginning and maintaining it globally (since the adjacency between two remaining vertices is non-decreasing in the recursion depth).} to keep the number of edges in the current graph appropriately small (cf.~Claim~\ref{cl:edges}),
which does not affect the recursive call {\sc SUP}$[\tG, \cdots]$.

The following claim completes the analysis.

\begin{claim}\label{cl:edges}
For any graph $G' = (V', E')$ that appears as an input of a recursive call of {\sc SUP},
if $G'$ has at most two parallel edges between each pair of two vertices, then we have $|E'| \leq 2|E|$.
\end{claim}

\begin{proof}
We prove $|E'| \leq 2|E|$ by constructing a mapping $\sigma \colon E' \to E$ such that $|\sigma^{-1}(e)| \leq 2$ for every $e \in E$.
Let $G = G_0, G_1, \dots, G_k = G'$ be the sequence of the $\Gamma$-labeled graphs such that $G_j = \tG_{j-1}$ $(j = 1, 2, \dots, k)$, i.e., 
{\sc SUP}$[G_j, \cdots]$ is recursively called in Step~5 of {\sc SUP}$[G_{j-1}, \cdots]$.
We define $\sigma$ as the composition of mappings $\sigma_j \colon E_j \to E_{j-1}$ $(j = 1, 2, \dots, k)$ defined as follows,
where $E_j$ denotes the edge set of $G_j$ for each $j = 0, 1, \dots, k$.

Let $e \in E_j$. If $e \in E_{j-1}$, then $\sigma_j(e) \coloneqq e$.
Otherwise, $e$ is a new edge $\tf_i = \{b, x\} \in E_j \setminus E_{j-1}$ $(i \in \{1, 2\})$
corresponding to some removed edge $f = \{w, x\} \in E_{j-1} \setminus E_j$,
and define $\sigma_j(\tf_i) = f$.

Let $\sigma \coloneqq \sigma_1 \circ \sigma_2 \circ \cdots \circ \sigma_k$.
For any edge $e = \{u, v\} \in E = E_0$,
each edge $e' \in \sigma^{-1}(e) \subseteq E' = E_k$ connects the same pair of two vertices $u', v' \in V'$ by the above definition,
where $u'$ and $v'$ are the vertices into which $u$ and $v$, respectively, are virtually merged by shrinking operations.
Thus, since each pair of two vertices has at most two parallel edges in $G'$,
we conclude that $|\sigma^{-1}(e)| \leq 2$ holds for every $e \in E$.
\end{proof}

\section{An $\mathrm{O}(m \log n)$-Time Algorithm without Explicit Shrinking}\label{sec:fast}
In this section, we present a faster algorithm for the shortest unorthodox path problem by avoiding explicitly shrinking blossoms, which runs in $\mathrm{O}(m \log n)$ time and completes the proof of Theorem~\ref{thm:main}.
The correctness of the algorithm is guaranteed based on a dual LP formulation (and, in fact, vice versa).

\subsection{Algorithm Description}
In this section, we present an $\mathrm{O}(m \log n)$-time algorithm for computing the shortest unorthodox path distance $q \in \RR^V \cup \{+\infty\}$ from $s$ to each vertex $v \in V$, where $q(v) = +\infty$ means that $G$ contains no unorthodox $s$--$v$ path.
We first describe the algorithm with a na\"ive implementation (which requires $\mathrm{O}(nm)$ time as with the recursive algorithm), and then explain that it can be implemented in $\mathrm{O}(m \log n)$ time with the aid of basic data structures (Lemma~\ref{lem:time}).
The correctness is shown in Lemma~\ref{lem:fast} via a dual LP formulation introduced in Section~\ref{sec:formulation}.
Following the case analysis in the proof of Lemma~\ref{lem:fast} (or the intuition of the auxiliary directed graph $\vec{H}$ given in Section~\ref{sec:formulation}),
one can reconstruct an unorthodox $s$--$t$ path $Q_t$ in $G$ with $\ell(Q_t) = q(t)$ in linear time.

The algorithm intuitively emulates the recursive algorithm {\sc SUP} without explicit blossom shrinking.
We instead maintain the information of into which vertex each vertex $v \in V$ is shrunk as $\parent(v)$,
and that of which edge $e \in E$ corresponds to an inconsistent edge forming a blossom in the current graph (from the viewpoint of the recursive algorithm) as a subset $F \subseteq E$ and its height $h(e)$ (which is doubled for convenience).
The shrinking structure can be nested, and $\root(v)$ denotes the vertex into which $v$ is shrunk in the current graph, which is obtained by repeating $v \leftarrow \parent(v)$ until $\parent(v) = v$.

We define $\depth_T(v) \coloneqq |E(P_v)|$ for each $v \in V$ and denote by $\pred_T(v)$ the predecessor of $v \in V \setminus \{s\}$ in $T$ (the vertex $u$ such that $P_v = P_u * (u, e, v)$ for some $e \in T$).
In Step~3.1, we choose an edge $e = \{u, v\} \in E$ corresponding to some inconsistent edge forming a lowest blossom $C$ in the current graph.
In Step~3.2, we trace back from $\root(u)$ and $\root(v)$ until their lowest common ancestor in the current $s$-SPT.
Then, in Step~3.3, the set $B$ corresponds to the vertex set of the lowest blossom $C$ minus its stem $w_1 = w_2$ (see Fig.~\ref{fig:C_e}).

\begin{algorithm*}[{{\sc FastSUP}$[G, \ell, s, T]$}]
\begin{description}
  \setlength{\itemsep}{0mm}
\item[]

\item[Input:]
  A connected $\Gamma$-labeled graph $G = (V, E)$, edge lengths $\ell \in \RR_{\geq 0}^E$, a vertex $s \in V$, and an $s$-SPT $T$ of $(G, \ell)$ with $\dist_T$ and $\psi_T$.

\item[Output:]
  The shortest unorthodox path distance $q \in \RR^V \cup \{+\infty\}$ from $s$.
  
\item[Step 1.]\vspace{1mm}
  Set $q(v) \leftarrow +\infty$ and $\parent(v) \leftarrow v$ for each $v \in V$.

\item[Step 2.]
  For each $e = \{u, v\} \in E$, set $h(e) \leftarrow \dist_T(u) + \dist_T(v) + \ell(e)$ if $e$ is inconsistent, and $h(e) \leftarrow +\infty$ otherwise.
  Initialize a set $F$ as the set of inconsistent edges.
  
\item[Step 3.]
  While $\Delta \coloneqq \min\left\{\, h(e) \mid e \in F \,\right\} < +\infty$, do the following.%\vspace{-1mm}
  \begin{description}
  \setlength{\itemsep}{.5mm}
    \item[Step 3.1.]
    	Pick $e = \{u, v\} \in F$ with $h(e) = \Delta$, update $F \leftarrow F \setminus \{e\}$, and set $w_1 \leftarrow \root(u)$ and $w_2 \leftarrow \root(v)$.
    \item[Step 3.2.]
      Set $B \leftarrow \emptyset$. While $w_1 \neq w_2$, let $i \in \{1, 2\}$ be an index with $\depth_T(w_i) \geq \depth_T(w_{3-i})$, update $B \leftarrow B \cup \{w_i\}$ and then $w_i \leftarrow \root(\pred_T(w_i))$.
    \item[Step 3.3.]
      For each $w \in B$,
      \begin{itemize}
      \setlength{\itemsep}{1mm}
      \item update $\parent(w) \leftarrow w_1 \ (= w_2)$,
      \item update $q(w) \leftarrow h(e) - \dist_T(w)$, and
      \item for each consistent edge $f = \{w, x\} \in \delta_G(w)$, if $h(f) > q(w) + \dist_T(x) + \ell(f)$, then update $h(f) \leftarrow q(w) + \dist_T(x) + \ell(f)$ and $F \leftarrow F \cup \{f\}$.
      \end{itemize}
  \end{description}
\end{description}
\end{algorithm*}

It can be observed that this algorithm also runs in $\mathrm{O}(nm)$ time.
Furthermore, with the aid of basic data structure, it can be implemented in $\mathrm{O}(m \log n)$ time as follows.

\begin{lemma}\label{lem:time}
  {\sc FastSUP}$[G, \ell, s, T]$ can be implemented in $\mathrm{O}(m \log n)$ time.
\end{lemma}

\begin{proof}
In order to maintain $F \subseteq E$ and $\root \in V^V$ (as well as $\parent \in V^V$), we utilize a priority queue with key $h$ and a disjoint-set forest (with path compression), respectively (see, e.g., \cite{CLRS2009}).

First, each vertex in $V$ is added to $B$ at most once as follows.
We only see $\root(v)$ (which remains in the current graph) for some $v \in V$ in Step~3.2 and we set $\parent(w) \leftarrow w_1 \neq w$ for each $w \in B$ in Step~3.3 (after that, $w$ cannot be $\root(v)$ for any $v \in V$).
Since each operation is performed in $\mathrm{O}(\log n)$ time (in the amortized sense) and all vertices appearing in Steps~3.1--3.2 but $w_1$ in Step~3.3 are added to $B$ (and then will never appear), the total cost for the disjoint-set forest is $\mathrm{O}(n \log n)$.

Next, each edge $e \in E$ is added to $F$ at most twice, because it happens only in Step~2 if $e$ is inconsistent, and only in Step~3.3 with one of its end vertices in $B$ if $e$ is consistent.
Since each operation (pop and push) is performed in $\mathrm{O}(\log m) = \mathrm{O}(\log n)$ time, the total cost for the priority queue is $\mathrm{O}(m \log n)$.
We remark that, if the key of $e$ is not equal to $h(e)$ in Step~3.1 (which may happen if $e$ is pushed twice), then nothing happens in the iteration because $\root(u) = \root(v)$ (some blossom with respect to the corresponding edge has already been shrunk).

The rest (initialization and update of $q$ and $h$) is simply done in linear time, and we are done.
\end{proof}

\subsection{Dual LP Formulation}\label{sec:formulation}
\begin{comment}
The unconstrained shortest path problem in a directed graph $\vec{G} = (V, \vec{E})$ admits an LP formulation
\begin{align*}%\label{eq:primal_LP_unconstrained}
  \begin{array}{rl}
  \mathrm{minimize} & \displaystyle\sum_{\vec{e} \in \vec{E}} \ell(\vec{e}) \cdot x(\vec{e})\\
  \mathrm{subject~to} & \displaystyle\sum_{\vec{e} \in \delta_{\vec{G}}^\mathrm{in}(v)} x(\vec{e}) - \sum_{\vec{e} \in \delta_{\vec{G}}^\mathrm{out}(v)} x(\vec{e}) = \begin{cases}
  -1 & (v = s),\\
  1 & (v = t),\\
  0 & (\text{otherwise}),
  \end{cases} \quad (v \in V)\\[8mm]
  & x \in \RR^{\vec{E}}_{\geq 0},
  \end{array}
\end{align*}
where $\delta_{\vec{G}}^\mathrm{in}(v)$ and $\delta_{\vec{G}}^\mathrm{out}(v)$ denote the sets of directed edges entering or leaving $v$, respectively.
\end{comment}
The unconstrained shortest path problem in a directed graph $\vec{G} = (V, \vec{E})$ admits a dual LP formulation as follows (see, e.g., \cite[$\S$~8.2]{Schrijver2003}):
\begin{align}\label{eq:dual_LP_unconstrained}
  \begin{array}{rl}
  \mathrm{maximize} & p(t) - p(s)\\[2mm]
  \mathrm{subject~to} & p(v) - p(u) \leq \ell(\vec{e}) \quad (\vec{e} = uv \in \vec{E}),\\[1mm]
  & p \in \RR^V,
  \end{array}
\end{align}
where the dual variable $p$ is called a \emph{potential}.
The undirected case admits an analogous formulation by replacing each edge with a pair of directed edges of opposite directions.
Then, for any $s$-SPT $T$ of $(G, \ell)$, the vector $\dist_T \in \RR^V$ gives an optimal potential.
In other words, Dijkstra's algorithm computes an optimal solution to \eqref{eq:dual_LP_unconstrained}.

We provide a similar formulation for the shortest unorthodox path problem, to which the algorithm {\sc FastSUP} computes an optimal solution.
Specifically, we construct an auxiliary directed graph $\vec{H}$ with vertex set $V \cup \{s^*\}$ and edge length $\ell^*$ such that some shortest $s^*$--$v$ path in $\vec{H}$ corresponds to a shortest unorthodox $s$--$v$ path in $G = (V, E)$ with edge length $\ell$.
We remark that, in the following construction, we just give an intuition without the correctness proof, which will be given by showing that the algorithm {\sc FastSUP}$[G, \ell, s, T]$ indeed solves the dual LP (Lemma~\ref{lem:fast}).

Let us fix an $s$-SPT $T$ of $(G, \ell)$ and introduce a new vertex $s^*$ as a supersource.
By the definition of blossoms (cf.~Definition~\ref{def:canonical}), each inconsistent edge $e = \{u, v\} \in E$ (as well as a blossom $C_e$ with base $b_e$) gives an unorthodox $s$--$w$ path $Q_w$ with $E(Q_w) \subseteq T \cup \{e\}$ for each $w \in V(P_u) \triangle V(P_v) = V(C_e) \setminus \{b_e\}$, where $X \triangle Y$ denotes the \emph{symmetric difference} $(X \setminus Y) \cup (Y \setminus X)$ of two sets $X$ and $Y$.
Thus, we create a directed edge $\vec{e}_w = s^*w$ with length $\ell^*(\vec{e}_w) = \ell(Q_w) = \dist_T(u) + \dist_T(v) + \ell(e) - \dist_T(w)$.

Let $e = \{u, v\} \in E$ be a consistent edge.
For each $w \in V(P_v) \setminus V(P_u)$, if $G$ has an unorthodox $s$--$u$ path $Q_u$ with $V(Q_u) \cap V(P_v[w, v]) = \emptyset$, then $Q_u$ can be extended to an unorthodox $s$--$w$ path $Q_w$ by the path $(u, e, v) * \overline{P_v}[v, w]$ (see Fig.~\ref{fig:LP1}).
Thus, we create a directed edge $\vec{e}_w = uw$ with length $\ell^*(\vec{e}_w) = \ell(Q_w) - \ell(Q_u) = \dist_T(v) + \ell(e) - \dist_T(w)$.
For each $w \in V(P_u) \setminus V(P_v)$, if $G$ has an unorthodox $s$--$u$ path $Q_u$ with $w \in V(Q_u)$, $Q_u[s, w] = P_w$, and $V(Q_u[w, u]) \cap V(P_v) = \emptyset$, then $P_v$ can be extended to an unorthodox $s$--$w$ path $Q_w$ by the path $(v, e, u) * \overline{Q_u}[u, w]$ (see Fig.~\ref{fig:LP2}).
Thus, we create a directed edge $\vec{e}_w = uw$ with length $\ell^*(\vec{e}_w) = \ell(Q_w) - \ell(Q_u) = \dist_T(v) + \ell(e) - \dist_T(w)$.
Note that the lengths are the same in both cases, and hence we do not need to distinguish them (just take $w \in V(P_u) \triangle V(P_v)$).
We also remark that the validity of these edges is unclear at this point because they may yield some undesired $s$--$w$ walk (not path) when $Q_u$ does not satisfy the assumptions.
This concern, however, turns out to be unfounded in Lemma~\ref{lem:fast}.

\begin{figure}[tb]
   \begin{center}
    \includegraphics[scale=0.8]{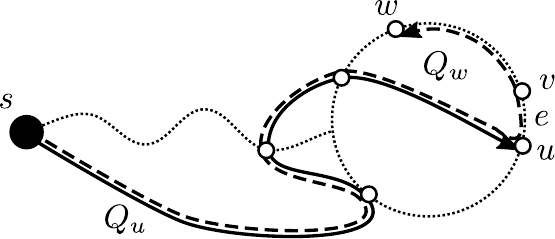}
   \end{center}\vspace{-3mm}
   \caption{A consistent edge $e = \{u, v\}$ extends an unorthodox $s$--$u$ path $Q_u$ (solid) to an unorthodox $s$--$w$ path $Q_w$ (dashed) for $w \in V(P_v) \setminus V(P_u)$, where the dotted lines represents $T \cup \{e\}$.}\vspace{7mm}
   \label{fig:LP1}
\end{figure}

\begin{figure}[tb]
 \begin{center}
  \includegraphics[scale=0.8]{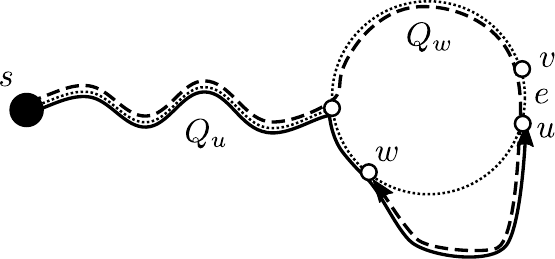}
 \end{center}\vspace{-3mm}
 \caption{Along an unorthodox $s$--$u$ path $Q_u$ (solid), a consistent edge $e = \{u, v\}$ extends $P_v$ to an unorthodox $s$--$w$ path $Q_w$ (dashed) for $w \in V(P_u) \setminus V(P_v)$, where the dotted lines represents $T \cup \{e\}$.}%\vspace{-1mm}
 \label{fig:LP2}
\end{figure}

We then obtain a shortest path instance $(\vec{H}, \ell^*)$, whose dual LP is as follows:
\begin{align}\label{eq:LP_unorthodox}
  \begin{array}{rll}
  \mathrm{maximize} & q(t) \\[2mm]
  \mathrm{subject~to} & q(w) \leq \dist_T(u) + \dist_T(v) + \ell(e) - \dist_T(w) & \left(\begin{array}{l}
    e = \{u, v\} \in E \colon \mathrm{inconsistent}\\
    w \in V(P_u) \triangle V(P_v)\end{array}\right),\\[4mm]
  & q(w) \leq q(u) + \dist_T(v) + \ell(e) - \dist_T(w) & \left(\begin{array}{l}
    e = \{u, v\} \in E \colon \mathrm{consistent}\\
    w \in V(P_u) \triangle V(P_v)\end{array}\right),\\[1mm]
    % & q(s^*) = 0,\\[1mm]
    & q \in \mathbb{R}^{V},
  \end{array}
\end{align}
where we slightly modify the original formulation \eqref{eq:dual_LP_unconstrained} by assuming $q(s^*) = 0$ without loss of generality.

In what follows, we show the correctness of \eqref{eq:LP_unorthodox} in the sense that, for any optimal solution $q \in \RR^V$ and any shortest unorthodox $s$--$t$ path $Q_t$ in $G$, we have $q(t) = \ell(Q_t)$ (also, \eqref{eq:LP_unorthodox} is unbounded if and only if $G$ contains no unorthodox $s$--$t$ path).
This also means that the shortest unorthodox path problem can be reduced to the unconstrained shortest path problem in directed graphs with $\mathrm{O}(n)$ vertices and $\mathrm{O}(nm)$ edges (at least if we only require the optimal value without the reconstruction of an unorthodox $s$--$t$ path attaining it).
In particular, the above argument provides a totally dual integral LP formulation of the shortest unorthodox path problem.

We first observe the weak duality, which intuitively implies that the auxiliary graph $\vec{H}$ has enough directed edges.
The strong duality, which means that the second-type directed edges (corresponding to consistent edges) are valid, is confirmed via the algorithm {\sc FastSUP}.

\begin{lemma}\label{lem:weak_dual}
  For any feasible solution $q \in \RR^V$ to \eqref{eq:LP_unorthodox}, any vertex $w \in V \setminus \{s\}$, and any unorthodox $s$--$w$ path $R_w$ in $G$, we have $q(w) \leq \ell(R_w)$.
\end{lemma}

\begin{proof}
Fix an arbitrary feasible solution $q \in \RR^V$ to \eqref{eq:LP_unorthodox}, and we prove this lemma by induction on $\mu(R_w) \coloneqq n\cdot|E(R_w) \setminus T| + |E(R_w)|$.
Let $e = \{u, v\}$ be the last edge on $R_w$ such that $w \in V(P_v)$ but $w \not\in V(P_u)$ (such an edge always exists since $w \in V(P_w)$ and $w \not\in V(P_s) = \{s\}$).
We separately consider the two cases when $e$ is inconsistent and when $e$ is consistent.

\medskip\noindent
\textbf{Case~1.}~
Suppose that $e$ is inconsistent.
Then, by Lemma~\ref{lem:tree}(1) and the first constraint in \eqref{eq:LP_unorthodox}, we have
\begin{align*}
  \ell(R_w) &= \ell(R_w[s, u]) + \ell(e) + \ell(R_w[v, w]) \geq \dist_T(u) + \ell(e) + \dist_T(v) - \dist_T(w) \geq q(w).
\end{align*}

\smallskip\noindent
\textbf{Case~2.}~
Suppose that $e$ is consistent.
From the choice of $e$, we have $w \in V(P_x)$ for every $x \in V(R_w[v, w])$, and hence $P_w * \overline{R_w}[w, v]$ is a path.
Then, since $R_w$ is unorthodox, so is at least one of $R_w[s, u]$ and $P_w * \overline{R_w}[w, v]$.

\medskip\noindent
\textbf{Case~2.1.}~
Suppose that $R_w[s, u]$ is unorthodox, and then $u \neq s$.
As $\mu(R_w[s, u]) < \mu(R_w)$, we have $q(u) \leq \ell(R_w[s, u])$ by the induction hypothesis.
Hence, by Lemma~\ref{lem:tree}(1) and the second constraint in \eqref{eq:LP_unorthodox}, we have
\begin{align*}
  \ell(R_w) &= \ell(R_w[s, u]) + \ell(e) + \ell(R_w[v, w]) \geq q(u) + \ell(e) + \dist_T(v) - \dist_T(w) \geq q(w).
\end{align*}

\smallskip\noindent
\textbf{Case~2.2.}~
Suppose that $P_w * \overline{R_w}[w, v]$ is unorthodox, and then $v \neq w$.
Since $w \not\in V(R_w[s, v])$ and $w \in V(P_v)$, we have $R_w[s, v] \neq P_v$, which implies $E(R_w[s, v]) \setminus T \neq \emptyset$.
As $E(P_w) \subseteq T$, we have $\mu(P_w * \overline{R_w}[w, v]) < \mu(R_w)$, and hence $q(v) \leq \ell(P_w * \overline{R_w}[w, v])$ by the induction hypothesis.
Hence, by Lemma~\ref{lem:tree}(1) and the second constraint in \eqref{eq:LP_unorthodox}, we have
\begin{align*}
  \ell(R_w) &= \ell(R_w[s, u]) + \ell(e) + \ell(R_w[v, w])\\
  &= \ell(R_w[s, u]) + \ell(e) + \ell(P_w * \overline{R_w}[w, v]) - \dist_T(w)\\
  &\geq \dist_T(u) + \ell(e) + q(v) - \dist_T(w) \geq q(w). \qedhere
\end{align*}
\end{proof}

\subsection{Correctness}
We prove the correctness of the algorithm {\sc FastSUP} and the dual LP formulation \eqref{eq:LP_unorthodox} as follows.
Recall that the weak duality has already been confirmed in Lemma~\ref{lem:weak_dual}.

\begin{lemma}\label{lem:fast}
  {\sc FastSUP}$[G, \ell, s, T]$ computes an optimal solution $q$ to $\eqref{eq:LP_unorthodox}$, which satisfies $q(t) = \ell(Q_t)$ for some unorthodox $s$--$t$ path $Q_t$ in $G$.
\end{lemma}

\begin{proof}
Let $k$ be the number of iterations of Step~3, and let $e_i$ denote the edge picked in the $i$-th iteration for each $i = 1, 2, \dots, k$.
When a consistent edge $f = \{w, x\}$ is added to $F$ in Step~3.3, we have
\begin{align*}
  q(w) + \dist_T(x) + \ell(f) = h(e) - \dist_T(w) + \dist_T(x) + \ell(f) \geq h(e).
\end{align*}
Since $h(f) \geq h(e)$ holds if $f$ is already in $F$, the sequence $h(e_1), h(e_2), \dots, h(e_k)$ is non-decreasing, where $h(e_i)$ denotes the value when $e_i$ is picked in Step~3.1.

For $i = 0, 1, \dots, k$, let $E_i \coloneqq \{e_1, e_2, \dots, e_i\}$, and let $\root_i$ denote $\root \in V^V$ after the $i$-th iteration.
We prove that the following properties are preserved by induction on $i$, where $q(w)$ is also the value after the $i$-th iteration:
\begin{itemize}
\setlength{\itemsep}{.5mm}
\item[1.] $q(w) = +\infty$ if and only if $\root_i(w) = w$;
\item[2.] for each $w \in V$, the set $\comp_i(w) \coloneqq \{\, x \in V \mid \root_i(x) = \root_i(w) \,\}$ induces a connected subgraph in $T$;
\item[3.] for each $b \in V$ with $\root_i(b) = b$, we have $\argmin\left\{\, \depth_T(w) \mid w \in \comp_i(b) \,\right\} = \{b\}$;
\item[4.] $q$ satisfies the constraints in \eqref{eq:LP_unorthodox} for all edges in $E_i$;
\item[5.] for each $w \in V$ with $q(w) \neq +\infty$, there exists an unorthodox $s$--$w$ path $Q_w$ in $G$ such that $\ell(Q_w) \leq q(w)$ and $E(Q_w) \subseteq T \cup E[\comp_i(w)]$, where $E[X]$ denotes the set of edges induced by a vertex set $X \subseteq V$.
\end{itemize}

Suppose that these properties hold at the end of the algorithm {\sc FastSUP}$[G, \ell, s, T]$.
As $F$ is empty at last, we have $q(u) = q(v) = +\infty$ for every edge $e = \{u, v\} \in E \setminus E_k$.
Hence, $q$ is feasible to \eqref{eq:LP_unorthodox}, which completes the proof by the property 5 (recall that the weak duality has already been observed in Lemma~\ref{lem:weak_dual}).

Initially, we have $E_0 = \emptyset$ and $\root_0(v) = v$ for every $v \in V$, and all the properties trivially hold.
Fix $i = 1, 2, \dots, k$, suppose that the properties 1--5 hold for $i - 1$, and let $e_i = \{u, v\}$.
The property~1 holds after the $i$-th iteration as we update $q(w)$ if and only if $w \in B$ in Step~3.3, and then $\parent(w)$ is updated simultaneously, which implies $\root_i(w) \neq w$ (vice versa).
The property~3 also holds by the update rule of $\parent$ in Step~3.3 since $\depth_T(u) < \depth_T(v)$ holds whenever $u = \pred_T(v)$ (by definition) or $u = \root_{i-1}(v)$ (by the induction hypothesis).

The property~2 is seen as follows.
If $\root_i(x) = \root_{i-1}(x)$ for every $x \in \comp_i(w)$, then $\comp_i(w) = \comp_{i-1}(w)$ and we are done by the induction hypothesis.
Otherwise, by symmetry and transitivity, it suffices to show that,  for any vertex $x \in \comp_i(w)$ with $\root_i(x) \neq \root_{i-1}(x)$, all the vertices on the path between $x$ and $\root_i(x)$ in $T$ are in $\comp_i(w)$.
If $\root_i(x) \neq \root_{i-1}(x)$, then we have $\root_{i-1}(x) \in B$ in Step~3.3, and $\root_i(x) = \parent(\root_{i-1}(x))$ (after Step~3.3) is found by tracing $\pred_T$ and $\root_{i-1}$ from $x$.
For any $y \in V \setminus \{s\}$, by definition, $\pred_T(y)$ is adjacent to $y$ in $T$, and all the vertices on the path between $y$ and $\root_{i-1}(y)$ in $T$ are in $\comp_{i-1}(y)$ (which induces a connected subgraph in $T$ by induction hypothesis).
Thus we conclude that all the vertices on the path between $x$ and $\root_i(x)$ in $T$ are in $\comp_i(\root_i(x)) = \comp_i(w)$ (recall that we update $\parent(y) = \root_i(x)$ for all the vertices $y \in B$ that appear in tracing back).

For the property~4, we have to care for a new constraint
\begin{align}\label{eq:new}
  q(w) \leq h(e_i) - \dist_T(w)
\end{align}
for each $w \in V(P_u) \triangle V(P_v)$.
This is because $h(e_i) = \dist_T(u) + \dist_T(v) + \ell(e_i)$ if $e_i$ is inconsistent, and otherwise $h(e_i) = \min\left\{q(u) + \dist_T(v) + \ell(e_i),\, q(v) + \dist_T(u) + \ell(e_i)\right\}$ (where $q(u)$ and $q(v)$ are their final values) as $e_i$ has the minimum priority $h(e_i)$ and the sequence $h(e_1), h(e_2), \dots, h(e_k)$ is non-decreasing.
As we set $q(w) \leftarrow h(e_i) - \dist_T(w)$ for each $w \in B$, the inequality \eqref{eq:new} holds for every $w \in B$.
If $w \in (V(P_u) \triangle V(P_v)) \setminus B$, then $\root_{i-1}(w) \neq w$ and hence $q(w) \neq +\infty$ by the induction hypothesis.
In particular, we have $q(w) = h(e_j) - \dist_T(w)$ for some $j < i$.
Since $h(e_j) \leq h(e_i)$, the inequality \eqref{eq:new} holds also in this case.
%By the properties 1--3 (induction hypothesis), we can compute $B = \{\, w \in V(P_u) \triangle V(P_v) \mid q(w) = +\infty \,\}$ in ${\rm O}(|B| \log n)$ time in Step~3.2 by jumping between roots in $\cF_{i-1}$, and hence this part is not a bottleneck (as it is ${\rm O}(n \log n)$ in total).
Thus we have confirmed the property 4.

In what follows, we prove the property 5.
We fix $w \in B$, and do the following case analysis, which is analogous to that in the proof of Lemma~\ref{lem:weak_dual} (see also Figs.~\ref{fig:LP1} and \ref{fig:LP2} for Case~2).

\medskip\noindent
\textbf{Case~1.}~
Suppose that $e_i$ is inconsistent.
If $w \in V(P_v) \setminus V(P_u)$, then we define $Q_w \coloneqq P_u * (u, e_i, v) * \overline{P_v}[v, w]$, so that $Q_w$ is indeed unorthodox and we have $\ell(Q_w) = \dist_T(u) + \ell(e_i) + \dist_T(v) - \dist_T(w) = q(w)$.
Otherwise, $w \in V(P_u) \setminus V(P_v)$, and we are done by symmetry.

\medskip\noindent
\textbf{Case~2.}~
Suppose that $e_i$ is consistent.
In this case, we cannot regard $u$ and $v$ as symmetric.
Suppose that $h(e_i) = q(u) + \dist_T(v) + \ell(e_i)$, and let $Q_u$ be an unorthodox $s$--$u$ path in $G$ such that $\ell(Q_u) \leq q(u)$ and $V(Q_u) \subseteq T \cup E[\comp_{i-1}(u)]$ by induction hypothesis.

\medskip\noindent
\textbf{Case~2.1.}~
Suppose that $w \in V(P_v) \setminus V(P_u)$.
We then define $Q_w \coloneqq Q_u * (u, e_i, v) * \overline{P_v}[v, w]$.
From the properties 1--3 (induction hypothesis), we see that $w$ is the root of $\comp_{i-1}(w)$ (since $q(w) = +\infty$ before Step~3.3), which separates $V(P_v[w, v])$ from $\comp_{i-1}(u) \neq \comp_{i-1}(w)$ in the tree $T$ (since $w \not\in V(P_u)$).
This implies
\begin{align*}
  V(P_v[w, v]) \cap V(Q_u) &\subseteq \bigl(V(P_v[w, v]) \cap V(P_u)\bigr) \cup \bigl(V(P_v[w, v]) \cap \comp_{i-1}(u)\bigr) = \emptyset.
\end{align*}
Thus, $Q_w$ is indeed an unorthodox path with $\ell(Q_w) \leq q(u) + \ell(e_i) + \dist_T(v) - \dist_T(w) = q(w)$.

\medskip\noindent
\textbf{Case~2.2.}~
Suppose that $w \in V(P_u) \setminus V(P_v)$.
In this case, we have $w \in V(Q_u)$ and $Q_u[s, w] = P_w$ as follows.
Suppose to the contrary that this is not the case.
Then, $Q_u$ leaves $P_u$ at some vertex $b \in V(P_w) \setminus \{w\}$, and let $f = \{b, x\} \in E(Q_u) \setminus E(P_u)$ be the first edge on $Q_u$.
From the property~5 (induction hypothesis), $f$ is induced by $\comp_{i-1}(u)$ and hence $b \in \comp_{i-1}(u)$.
As $u \in \comp_{i-1}(u)$ by definition, by the properties~2 and 3 (induction hypothesis), we must have $w \in \comp_{i-1}(u) \setminus \{\root_{i-1}(u)\}$, which contradicts the property 1 (induction hypothesis).

We then define $Q_w \coloneqq P_v * (v, e_i, u) * \overline{Q_u}[u, w]$, and as $w \not\in V(P_v)$, we see that
\begin{align*}
  V(P_v) \cap V(Q_u[w, u]) &\subseteq \bigl(V(P_v) \cap V(P_u[w, u])\bigr) \cup \bigl(V(P_v) \cap (\comp_{i-1}(u) \setminus \{\root_{i-1}(u)\})\bigr) = \emptyset.
\end{align*}
Thus, $Q_w$ is indeed an unorthodox path with $\ell(Q_w) \leq \dist_T(v) + \ell(e_i) + q(u) - \dist_T(w) = q(w)$.

\medskip
Finally, in all cases, we have 
\begin{align*}
  E(Q_w) &\subseteq \bigl(T \cup E[\comp_{i-1}(u)]\bigr) \cup \bigl(T \cup E[\comp_{i-1}(v)]\bigr) \cup \{e_i\} \subseteq T \cup E[\comp_{i}(w)],
\end{align*}
where the last inclusion holds because $w \in B$ has the same root as both $u$ and $v$ (after Step~3.3).
Thus we are done.
\end{proof}

\begin{remark}
The algorithm {\sc FastSUP}$[G, \ell, s, T]$ computes the lengths of shortest unorthodox $s$--$w$ paths in $G$ for all possible vertices $w \in V \setminus \{s\}$ at once, and each shortest unorthodox $s$--$w$ path $Q_w$ is reconstructed in linear time by following the case analysis in the proof of Lemma~\ref{lem:fast}.
Also, the recursive algorithm {\sc SUP}$[G, \ell, s, t, T]$ can be modified so, e.g., by virtually setting $t = s$ until it returns ``NO'' (i.e., all possible vertices have been shrunk) and by remembering shortest unorthodox $s$--$w$ paths expanded into the original graph for all shrunk vertices $w \in V(C_e) \setminus \{b_e\}$ in Step~3 (which are utilized for expanding operations in further recursive calls).
\end{remark}

\section*{Acknowledgments}
The authors are deeply grateful to the anonymous reviewers of this paper and the preliminary version~\cite{Yamaguchi2020} for their valuable comments and suggestions.
This work was partially supported by RIKEN Center for Advanced Intelligence Project.

\end{document}